\documentclass{amsart}
\usepackage{amsmath}
\usepackage{amsfonts}
\usepackage{amssymb}
\usepackage{graphicx}
\usepackage{color}
\usepackage[T1]{fontenc}
\usepackage[utf8]{inputenc}
\newtheorem{Theorem}{Theorem}[section]

\newtheorem{Lemma}[Theorem]{Lemma}
\newtheorem{Proposition}[Theorem]{Proposition}

\newtheorem{Definition}[Theorem]{Definition}
\newtheorem{rem}[Theorem]{Remark}
\newtheorem{example}[Theorem]{Example}

\newcommand{\R}{\mathbb R}
\newcommand{\K}{\mathbb{K}}
\newcommand{\N}{\mathbb N}
\newcommand{\Z}{\mathbb Z}

\newcommand{\C}{\mathbb{C}}
\newcommand{\mbH}{\mathbb H}
\newcommand{\F}{\mathcal{F}}

\newcommand{\D}{\mathcal{D}}

\title[Realization of quaternionic KP hierarchy]{On the Kadomtsev-Petviashvili hierarchy in an extended class of formal pseudo-differential operators}
\author{Jean-Pierre Magnot$^1$}

\author{Vladimir Roubtsov$^{2,3,4}$}
\address{$1$ and $2:$ LAREMA, Universit\'e d'Angers, \\ 2 boulevard Lavoisier, \\ 49045 Angers Cedex 1, France}
\address{$1$ Lyc\'ee Jeanne d'Arc, \\ Avenue de Grande Bretagne, \\ 63000 Clermont-Ferrand, France}
\address{$3$ ITEP, Theoretical Division, 25, Bol. Tcheremushkinskaya, 117259, Moscow, Russia}
\address{$4$ IGAP (Institute for Geometry and Physics), Trieste, Italy}
\email{$1:$ jean-pierr.magnot@ac-clermont.fr}
\email{$1:$ magnot@math.cnrs.fr}
\email{$2:$ volodya@univ-angers.fr}
\email{$2:$ volodya@math.cnrs.fr}

\begin{document}

	\begin{abstract}
		 We study the existence and uniqueness of the Kadomtsev-Petviashvili (KP) hierarchy solutions in the algebra of $\F Cl(S^1,\K^n)$ of formal classical pseudo-differential operators. The classical algebra $\Psi DO(S^1,\K^n)$ where the KP hierarchy is well-known appears as a subalgebra of $\F Cl(S^1,\K^n).$ We investigate algebraic properties of $\F Cl(S^1,\K^n)$ such as splittings, r-matrices, extension of the Gelfand-Dickii bracket, almost complex structures. Then, we prove the existence and uniqueness  of the KP hierarchy solutions in  $\F Cl(S^1,\K^n)$ with respect to extended classes of initial values. Finally, we extend this KP hierarchy to complex order formal pseudo-differential operators and we describe their Hamiltonian structures similarly to previously known formal case. .       
	\end{abstract}
	
	\maketitle
	
	\section*{Introduction}
	In the classical theory of the Kadomtsev-Petviashvili (KP) hierarchy, the considering algebra of pseudo-operators  is $$ \Psi DO(S^1,\K) = C^\infty(S^1,\K)((\partial^{-1}))$$ where $\partial$ is a derivation and  $\K = \R, \C$ or $\mathbb{H}.$ Classically, $\partial = \frac{d}{dx}, \, x\in S^1.$ 
	It is well-known that this KP hierarchy is an integrable system, with existence and uniqueness of solutions with respect to a fixed initial value, (
see e.g. \cite{D} for a classical treatise). There exists various generalizations, or deformations, of the KP hierarchy which almost all satisfy the formal integrability condition, and solutions satisfy properties similar to the properties of the solutions of the (classical) KP hierarchy. Recently, well-posedness have been stated for these equations \cite{ERMR,MR2016}. Classical algebraic settings that arise in the theory of the KP hierarchy will be reviewed in section \ref{ss:Lie}, in section \ref{s:poisson} and in section \ref{ss:prelKP}. 
	
	Pseudo-differential operators appear also in some contexts other than the theory of integrable systems. In general, larger classes of such operators are studied, see e.g. \cite{Gil,PayBook,Scott,See}, starting from \textbf{non-formal} operators, i.e. operators acting on spaces of sections of a vector bundle. These non-formal operators, in particular \textbf{classical} pseudo-differential operators, have their own applications and one can build from them spaces of formal classical operators. 
	The algebra of operators that we intend to use in this paper is the algebra of formal classical pseudo-differential operators $\F Cl(S^1,\K^n)$ that are obtained from classical pseudo-differential operators 
	acting on smooth sections of the trivial vector bundle $S^1 \times \K^n$ over $S^1,$ for $K =  \C$ or $\mathbb{H},$ see e.g. \cite{Gil,PayBook}. 
	In this algebra, it is possible to define functions of an elliptic positive operator that satisfy mild properties of the spectrum  using a Cauchy-like formula \cite{PayBook,Scott,See}. 
	In particular the square root of the Laplacian $|D| = \Delta^{1/2}$ is well-defined, as well as the sign of the Dirac operator $D = i \frac{d}{dx}$ defined by $$\epsilon(D) = D |D|^{-1} =  |D|^{-1} D.$$
	This operator is not in $ \Psi DO(S^1,\K^n).$ In fact, the algebra $ \Psi DO(S^1,\K)$ is the formal part of the so-called even-even class of (non-formal) classical pseudo-differenbtial operators first defined, to our knowledge, by Kontsevich and Vishik \cite{KV1,KV2} and named as even-even class operators in \cite{PayBook,Scott}, mostly motivated by problems about renormalized determinants. As a consequence, $ \Psi DO(S^1,\K)$ is a subalgebra of $\F Cl(S^1,\K^n)$ { which is noted in \cite{PayBook,Scott} as $\mathcal{F} Cl_{ee}(S^1,\K^n).$} The necessary properties  of these  pseudo-differential operator algebras, both formal and non-formal, will be reviewed in section \ref{ss:prel}. { The key properties of $\epsilon(D)$ that we use in our constructions are:
		
		\begin{itemize}
			\item the formal operator $\epsilon(D)\in \mathcal{F} Cl(S^1,\K^n)$ commutes with any formal operator $A \in  \mathcal{F} Cl(S^1,\K^n),$
			\item $\epsilon(D)^2=Id$
			\item the composition on the left $A \mapsto \epsilon(D) \circ A$ is an endomorphism of the algebra $ \mathcal{F} Cl(S^1,\K^n),$ which restricts to a bijiective map from $\Psi DO(S^1,\K^n)=\F Cl_{ee}(S^1,\K^n)$ to an algebraic complement in $ \mathcal{F} Cl(S^1,\K^n)$ noted as $ \mathcal{F} Cl_{eo}(S^1,\K^n)$ following the terminology of \cite{Scott}
			\item the restriction of the Wodzicki residue to $\Psi DO(S^1,\K^n)=\F Cl_{ee}(S^1,\K^n)$, which is similar to but not equal to the Adler functional, is vanishing.  
		\end{itemize}  
	}
	Our first remarks are the following: 
	\begin{itemize}
		\item The space $\F Cl(S^1,\K^n)$ splits in various ways: one is derived from the splitting of $T^*S^1 - S^1$ into two connected components (section \ref{ss:+-}), the splitting with respect to $\Psi DO(S^1,\K^n)$ as a subalgebra (section \ref{s:even/odd}), and the extension of the splitting related to the classical Manin triple on $\Psi DO(S^1,\K^n)$ to $\F Cl(S^1,\K^n)$ (section \ref{ss:manin}) . 
		\item The operator $\epsilon(D)$ is in the center of  $\F Cl(S^1,\K^n).$  It generates  then a polarized Lie bracket using it as a $\mathbf r-$matrix (section \ref{ss:polar}) and an integrable almost complex structure on  $\F Cl(S^1,\K^n).$ 
	\end{itemize}
These technical features enables us to state the announced main results of this paper: existence and uniqueness of solutions of the KP hierarchy with various initial conditions (section \ref{ss:multKP}) and KP hierarchy with complex powers (section \ref{ss:complex}). 

The paper is organized as follows:

Section \ref{s:tech} is devoted to technical preliminaries: we remind and review some operator algebras, Poisson structures and Manin pairs. We give an overview of the classical method for solving the KP hierarchy. New results of this Section are concentrated in section \ref{ss:complexpow} where formal operators of complex order that generalize operators in $\F CL(S^1,\K^n)$,  extending the definitions present in \cite{EKRRR1995}, \cite{KZ},  are described. In section \ref{s:Manin}, we explore some Manin pairs on $\F Cl(S^1\K^n)$, and in section \ref{ss:polar} we present some polarized brackets, inherited from the richer structure of $\F Cl(S^1\K^n)$.

Section \ref{s:inj} is focused on the comparison of $\F Cl(S^1\K^n)$ with  $\Psi DO(S^1,\K^n).$ First, we develop various injections of $\Psi DO(S^1,\K^n)$ in $\F Cl(S^1\K^n),$ beyond the standard one described in section \ref{s:tech}. Second, we describe three almost complex structures on $\F Cl(S^1\K^n)$ $J_1,$ $J_2$ and $J_3$ such that each couple $(J_1,J_2),$ $(J_1,J_3)$ and $(J_2,J_3)$ form an almost quaternionic structure on $\F Cl(S^1\K^n).$ We prove the integrability of $J_1,$ derived from $i\epsilon(D),$ and the non-integrability of the two others $J_2$ and $J_3.$

Section \ref{s:KP} deals with various type of initial values for the KP system, which are derived from the various injections of $\Psi DO(S^1,\K^n)$ in $\F Cl(S^1,\K^n),$ and ends up with a generalization  to the KP hierarchy with operators of complex order. As it was announced, the existence and uniqueness of the solutions, depending on the initial value, is stated. We make few short remarks about well-posedness.
 
The final part of the paper extends the classical Hamiltoinian formulations of the KP hierarchy from $\Psi DO(S^1,\K^n)$ to $\F CL(S^1,\K^n),$ using a generalized Adler-Gelfand-Dickii construction. 

All technical and routine proofs are gathered and organized in the Appendix.

\subsection{Acknowledgements.} This research of both authors was supported by LAREMA UMR 6093 du CNRS. V.R. was partly supported by the project IPaDEGAN (H2020- MSCA-RISE-2017), Grant Number 778010, and by the Russian Foundation for Basic Research under the Grants RFBR.
\subsection{ Conflict of Interest}  The authors declare that they have no conflicts of interest.

	\section{Technical preliminaries}\label{s:tech}
	\subsection{Preliminaries on pseudo-differential operators}\label{ss:prel}
	\subsubsection{Description}
	We {shall start with a description of (non-formal!) pseudo-differential operator groups and algebras which we consider in this work. Throughout this section $E$ denotes a} complex finite-dimensional vector bundle over $S^1.$
	 We {shall} specialize below  to the case $E = S^1 \times V$ in which $V$ is a $n-$dimensional vector space. 
	The following definition appears in \cite[Section 2.1]{BGV}.
	
	\begin{Definition} \label{def:diff-op}
		The graded algebra of differential operators acting on the space of smooth sections $C^\infty(S^1,E)$ is the 
		algebra $DO(E)$ generated by:
		
		$\bullet$ elements of $End(E),$ the group of smooth maps $E \rightarrow E$ leaving each fibre globally invariant  
		and which restrict to linear maps on each fibre. This group acts on sections of $E$ via (matrix) multiplication;
		
		$\bullet$  {covariant derivation} operators
		$$\nabla_X : g \in C^\infty(S^1,E) \mapsto \nabla_X g$$ where $\nabla$ 
		is a {smooth } connection on $E$ and $X$ is a {smooth} vector field on $S^1$.
	\end{Definition}
	
	{We assign as usual the order $0$  to smooth function} multiplication operators. 
	{The derivation} operators and vector fields have the  
	order 1. A differential operator of order $k$ has the form 
	$ P(u)(x) = \sum p_{i_1 \cdots i_r} \nabla_{x_{i_1}} \cdots \nabla_{x_{i_r}} u(x) \; , \quad r \leq k \; ,$
	In local coordinates (the coefficients $p_{i_1 \cdots i_r}$ can be matrix-valued).
	We denote by $DO^k(S^1)$,$k \geq 0$, the differential operators of order less or equal than $k$.
	The algebra $DO(E)$ is {filtered  by the} order. It is a subalgebra of the algebra of classical pseudo-differential operators $Cl(S^1,V)$ that we describe shortly hereafter, focusing on its necessary aspects.
	This is an algebra that contains, for example, the square root of the Laplacian \begin{equation} \label{eq:integral}|D| = \Delta^{1/2} = \int_{\Gamma} \lambda^{1/2}(\Delta-\lambda Id)^{-1} d\lambda,\end{equation}
	where $\Delta = - \frac{d^2}{dx^2}$ is the positive Laplacian and $\Gamma$ is a contour around the spectrum of the Laplacian, see e.g. \cite{See,PayBook} for an exposition on contour integrals of pseudo-differential operators. $Cl(S^1,V)$ contains also the inverse of $Id+\Delta,$ and all
	smoothing operators on $L^2(S^1,V). $ Among smoothing operators one can find the heat operator 
	$$e^{-\Delta} = \int_{\Gamma} e^{-\lambda}(\Delta-\lambda Id)^{-1} d\lambda.$$ 
	pseudo-differential operators (maybe non-scalar) are linear operators acting on $C^\infty(S^1,V)$ which reads locally as
	$$ A(f) = \int e^{ix.\xi}\sigma(x,\xi) \hat{f}(\xi) d\xi$$ where $\sigma \in C^\infty(T^*S^1, M_n(\C))$ satisfying additional estimates on its partial derivatives {and $\hat{f}$ means the Fourier transform of $f$}. Basic facts on pseudo-differential operators defined 
	on a vector bundle $E \rightarrow S^1$ can be found e.g. in \cite{Gil}.
	
{\begin{rem} Since $V$ is finite dimensional, there exists $n \in \N^*$ such that $V \sim \C^n.$ Through this identification, a pseudo-differential operator $A \in Cl^(S^1,V)$ can be identified with a matrix $(A_{i,j})_{(i,j)\in \N_n^2}$ with coefficients $$A_{i,j} \in Cl(S^1,\C).$$ In other words, the identification $V \sim \C^n$ that we fix induces the isomorphism of algebras
 $$Cl(S^1,V) \sim M_n(Cl(S^1,\C)).$$
 This identification will remain true and useful in the successive constructions below, and will be recalled if appropriate. When it will not carry any ambiguity, we will use the notation $DO(S^1),$ $Cl(S^1),$ etc. instead of $DO(S^1,\C),$ $Cl(S^1,\C),$ etc. for operators acting on the space of smooth functions from $S^1$ to $\C.$
\end{rem}}
	 {Pseudo-differential operators can be also described by their kernel $$K(x,y) = \int_{\R} e^{i(x-y)\xi} \sigma(x,\xi)d\xi$$ which is off-diagonal smooth.} Pseudo-differential operators {with infinitely smooth kernel (or "smoothing" operators)}, i.e. that are maps: $L^2 \rightarrow C^\infty$ form a two-sided ideal that we note by $Cl^{-\infty}(S^1,V).$ Their symbols  are those which are in the Schwartz space $\mathcal{S}(T^*S^1, M_n(\C)).$  The quotient 
	$\F Cl(S^1,V)=Cl(S^1,V)/Cl^{-\infty}(S^1,V)$ of the algebra of pseudo-differential operators by $Cl^{-\infty}(S^1,V)$ forms the algebra of formal pseudo-differential operators. 
	 {Another algebra, which is actually known as a subalgebra of $\F Cl(S^1,V)$ following \cite{MR2018}, is also called algebra of formal pseudo-differential operators. This algebra is generated by formal Laurent series $$ \Psi DO(S^1,V) = C^\infty(S^1,V)((\partial^{-1}))=\bigcup_{d \in \Z}\left\{ \sum_{k \leq d} a_k \partial^k \right\}$$ where each $a_k \in C^\infty(S^1,M_n(\C))$ and $\partial = \frac{d}{dx}.$ Let us precise hereafter a short but complete description of basic correspondence between $ \Psi DO(S^1,V)$ and $\F Cl(S^1,V)$.}
	
	Symbols $\sigma$ project to formal symbols and there is an isomorphism between formal pseudo-differntial operators and formal symbols. A detailed study can be found in \cite[Tome VII]{Dieu}.  Classical pseudo-differential operators are operators $A$ which associated formal symbol $\sigma(A)$ reads as an asymptotic expansion $$\sigma(A)(x,\xi) \sim \sum_{k \in \Z, k \leq o} \sigma_k(A)(x,\xi)$$
	where the \textbf{partial symbol of order k} $$\sigma_k(A): (x,\xi)\in T^*S^1\setminus S^1 \mapsto \sigma_k(A)(x,\xi)\in M_n(\C)$$ is $k-$positively homogeneous in the $\xi-$variable, smooth on $T^*S^1\setminus S^1 = \{(x,\xi)\in T^*S^1 \, | \, \xi \neq 0\}$ and such that {$d\in \Z$} is the order of the operator $A.$ The order of a smoothing operator {we put} equal to $-\infty$ and the formal symbol of a smoothing operator is $0.$

	 The set $\mathcal{F}Cl(S^1,V)$ is not the same as the space of formal operators $ \Psi DO(S^1,V) $ which naturally arises  in the algebraic theory of PDEs, see e.g. \cite{KW} for an overview, but here the partial symbols $\sigma_k(A)$ of $A \in \Psi DO(S^1,V)$ are $k-$homogeneous. By the way one only has $\Psi DO(S^1,V) \subset \mathcal{F}Cl(S^1,V).$ Following the remarks given in \cite{MR2018}, 
	 $\Psi DO(S^1,V)$ correspond to \textbf{even-even class} formal pseudo-differential operators { that we describe in section \ref{s:even/odd}.}
	Two approaches for a global symbolic calculus {of }pseudo-differential operators have been described in \cite{BK,Wid}. It  is shown in these papers how the 
	geometry of the base manifold $M$ furnishes an obstruction to generalizing 
	local formulas of of symbol composition and inversion; we do not recall these
	formulas here since they are not involved in our computations. 
	We assume henceforth (following e.g. \cite{Ma2016}, along the lines of the more general description of \cite{Gil}), that $S^1$ is equipped with charts such that the changes 
	of coordinates are translations.  Under these assumptions, 
	$$\quad \sigma(A \circ B) \sim \sum_{\alpha \in \N} \frac{(-i)^\alpha}{\alpha !} D^\alpha_\xi \sigma(A) D^\alpha_x \sigma(B), \quad \forall  A,B \in Cl(S^1,V), $$
	and specializing to partial symbols:
	$$\forall k \in \Z, \sigma(A \circ B)_k   = \sum_{\alpha \in \N}\sum_{m+n-\alpha = k} \frac{(-i)^\alpha}{\alpha !} D^\alpha_\xi \sigma_m(A) D^\alpha_x \sigma_n(B).$$
	{The composition $\sigma(A \circ B)$ for $A,B \in \Psi DO(S^1,V) \subset \mathcal{F}Cl(S^1,V)$ gives rise to a (unitary) associative algebra structure on $\Psi DO(S^1,V)$
	and we shall write in this case (by abuse of notation)
	\begin{equation}\label{asspsi}
	A \circ B = \sum_{\alpha \in \N} \frac{(-i)^\alpha}{\alpha !} D^\alpha_\xi A D^\alpha_x B
	\end{equation}}
	
      {\begin{rem}
	In such an  "operator product" we shall always suppose so called "Wick order" which means that we write functions on $C^{\infty}(S^1)$ on (or "in front of") left-hand side of all degrees of $D.$ 
	\end{rem}}
	\vskip 6pt
	\noindent
	\textbf{Notations.} 
	We shall denote note by
	$Cl^d(S^1,V)$ the vector space of classical pseudo-differential operators 
	of order{ $\leq d$}. We also denote by $Cl^{*}(S^1,V)$ the group of  {invertible in  $Cl(S^1,V)$ operators}.
	We denote the sets of formal operators adding the script $\mathcal{F}$. The algebra of formal pseudo-differential operators, identified wih formal symbols, is noted by ${\mathcal F}Cl^{}(S^1,V),$ and its group of invertible element is ${\mathcal F}Cl^{*}(S^1,V),$ while formal pseudo-differential operators of order less or equal to  $d \in \Z$ is noted by ${\mathcal F}Cl^{d}(S^1,V).$
	{%
\begin{rem}
	Through identification of $\F Cl(S^1,V)$ with the corresponding space of formal symbols, the space $\F Cl(S^1,V)$ is equipped with the natural locally convex topology inherited from the space of formal symbols. 
	A formal symbol $\sigma_k$ is a smooth function in  $C^\infty(T^*S^1\setminus S^1, M_n(\C))$ which is $k-$homogeneous (for $k>0)$), and hence with an element of $C^\infty(S^1, M_n(\C))^2$ evaluating 
	$\sigma_k$ at $\xi = 1$ and $\xi = -1.$ Identifyting $Cl^d(S^1,V)$ with 
	$$ \prod_{k \leq d} C^\infty(S^1, M_n(\C))^2, $$ 
	the vector space $Cl^d(S^1,V)$ is a Fr\'echet space, 
	and hence $$Cl(S^1,V) = \cup_{d \in \Z}Cl^d(S^1,V)$$ 
	is a locally convex topological algebra. 
	
	We have to precise that the classical topology on non-formal classical pseudo-differential operators $Cl(S^1,V)$ is finer than the one obtained by pull-back from $\F Cl(S^1,V).$ A ``useful'' topology on $Cl(S^1,V)$ needs to ensure that partial symbols \textbf{and} off-diagonal smooth kernels converge.  
	The topology on spaces of classical pseudo differential operators has been described by Kontsevich and Vishik in 
	\cite{KV1}; see also \cite{CDMP,PayBook,Scott} for descriptions.
	This is a Fr\'echet topology on each space 
	$Cl^d(S^1,E).$ 
	However, passing to the quotients
	$\F Cl^d(S^1,E) = Cl^d(S^1,E) /  Cl^{-\infty}(S^1,E),$
	the push-forward topology coincides with the topology of $\F Cl^d(S^1,V)$ described at the beginning of this remark. \end{rem}}
	
	

	\subsubsection{The splitting with induced by the connected components of $T^*S^1\setminus S^1.$} \label{ss:+-}
	In this section, we define two ideals of the algebra $\mathcal{F}Cl(S^1,V)$, 
	that we call $\mathcal{F}Cl_+(S^1,V)$ and $\mathcal{F}Cl_-(S^1,V)$, such that $\mathcal{F}Cl(S^1,V) = \mathcal{F}Cl_+(S^1,V) \oplus \mathcal{F}Cl_-(S^1,V)$. 
	This decomposition is explicit in \cite[section 4.4., p. 216]{Ka}, and we give an explicit description here following \cite{Ma2003,Ma2006-2}. 
	
	\begin{Definition}
		
		Let $\sigma$ be a partial symbol of order $o$ on $E$. Then, we define, for $(x,\xi) \in T^*S^1\setminus S^1$, 
		$$ \sigma_+(x,\xi) = \left\{ 
		\begin{array}{ll}
		\sigma(x,\xi) & \hbox{ if $ \xi > 0$} \\
		0 & \hbox{ if $ \xi < 0$} \\
		\end{array}
		\right. \hbox{ and }
		\sigma_-(x,\xi) = \left\{ 
		\begin{array}{ll}
		0 & \hbox{ if $ \xi > 0$} \\
		\sigma(x,\xi) & \hbox{ if $ \xi < 0$} . \\
		\end{array}
		\right.$$
		We define  $p_+(\sigma) = \sigma_+$ and $p_-(\sigma) = \sigma_-$ .
	\end{Definition}
	The maps 
	$ p_+ : \mathcal{F}Cl(S^1,V) \rightarrow \mathcal{F}Cl(S^1,V) $ { and } $p_- : \mathcal{F}Cl(S^1,V) \rightarrow \mathcal{F}Cl(S^1,V)$ are clearly smooth algebra morphisms (yet non-unital morphisms) that leave the order invariant and are also projections (since multiplication on formal symbols is expressed in terms of point-wise multiplication of tensors). 
	
	\begin{Definition} We define
		$  \mathcal{F}Cl_+(S^1,V) = Im(p_+) = Ker(p_-)$
		and $  \mathcal{F}Cl_-(S^1,V) = Im(p_-) = Ker(p_+).$ \end{Definition}
	Since $p_+$ is a projection,  we have the splitting
	$$ \mathcal{F}Cl(S^1,V) = \mathcal{F}Cl_+(S^1,V) \oplus \mathcal{F}Cl_-(S^1,V) .$$
	Let us give another characterization of $p_+$ and $p_-$. {The operator $D = {-i} \frac{d}{dx}$ splits $C^\infty(S^1, \C^n)$ into three spaces :
\begin{itemize}
\item  its kernel $E_0,$ built of constant maps
\item $E_+$, the vector space spanned by eigenvectors related to positive eigenvalues
\item $E_-$, the vector space spanned by eigenvectors related to negative eigenvalues.
\end{itemize}
	The $L^2-$orthogonal projection on $E_0$ is a smoothing operator, which has null formal symbol. By the way, concentrating our attention on thr formal symbol of operators, we can ignore this projection and hence we work on $E_+ \oplus E_-$. The following elementary result will be useful for the sequel.  
	\begin{Lemma} \label{l1} \cite{Ma2003,Ma2006-2}
\begin{itemize}	
\item $\sigma(D) = {\xi }, \quad \sigma(|D|) = {|\xi| }$
\item $\sigma(\epsilon) = {\xi \over |\xi|}$, where $\epsilon = D|D|^{-1} = |D|^{-1}D$ is the sign of D.
\item  Let $p_{E_+}$ (resp. $p_{E_-}$) be the projection on $E_+$ (resp. $E_-$), then $\sigma(p_{E_+}) ={1 \over 2}(Id + {\xi \over |\xi|})$ and $\sigma(p_{E_-}) = {1 \over 2}(Id - {\xi \over |\xi|})$.
\end{itemize}
\end{Lemma}
	

	Let us now give an easy but very useful lemma:
	\begin{Lemma}\label{l2} \cite{Ma2003}
		Let $f: \R^* \rightarrow V$ be a 0-positively homogeneous function with values in a topological vector space $V$. Then, for any $n \in \N^*$, $f^{(n)} = 0$ where $f^{(n)}$ denotes the n-th derivative of $f$.
	\end{Lemma}}
 
	From this, we have the following result.
	
	\begin{Proposition} \label{pag} \cite{Ma2003,Ma2006-2}
		Let $A \in \mathcal{F}Cl(S^1,V).$ 
		$ p_+(A) =  \sigma( p_{E_+}) \circ A = A \circ \sigma( p_{E_+})$ and 
		$  p_-(A) =  \sigma( p_{E_-}) \circ A = A \circ \sigma( p_{E_-}).$
	\end{Proposition}
	
	\noindent
	{\textbf{Notation.} For shorter notations, we note by $A_\pm = p_\pm(A)$ the formal operators defined from another viewpoint by $$\sigma (A_+)(x,\xi) \quad (\hbox{ resp. }\sigma (A_-)(x,\xi) ) = \left\{ \begin{array}{ll}
	\sigma (A)(x,\xi) & \hbox{if }\xi > 0 \quad(\hbox{ resp. }\xi < 0 ) \\
	0 & \hbox{if }\xi < 0  \quad(\hbox{ resp. }\xi > 0 )\\
	\end{array} \right.$$

	\subsubsection{The ``odd-even'' splitting} \label{s:even/odd}
	
	We note by $\sigma(A)(x,\xi)$ the total formal symbol of $A \in \mathcal{F}Cl(S^1,V).$ 
	The following proposition is trivial:
	\begin{Proposition}
		Let 
		$\phi:\mathcal{F} Cl(S^1,V)\rightarrow \mathcal{F}Cl(S^1,V)$ defined by 
		$$\phi(A) = \frac{1}{2}\sum_{k \in \mathbb{Z}}\sigma_{k}(A)(x,\xi) - (-1)^k\sigma_{k}(A)(x,-\xi).$$
		This map is smooth, and $\Psi DO(S^1,V)=\mathcal{F} Cl_{ee}(S^1,V) = Ker(\phi).$
	\end{Proposition}
	
	Following \cite{Scott}, one can define \textbf{even-odd class} pseudo-differential operators 
	$$\mathcal{F}Cl_{eo}(S^1,V) = \left\{ A \in\mathcal{F} Cl(S^1,V) \, | \, \sum_{k \in \mathbb{Z}}\sigma_{k}(A)(x,\xi) + (-1)^{k}\sigma_{k}(A)(x,-\xi) = 0 \right\}.$$
	\begin{rem}
		This terminology is inherited from \cite{Scott}. This reference is mostly concerned with non-formal operators. We have also to mention that the class of non formal  even-even pseudo-differential operators was first described in  \cite{KV1,KV2}. In these two references, even-even class pseudo-differential operators are called odd class pseudo-differential operators. By the way, following the terminology of \cite{KV1,KV2} even-odd class pseudo-differential operators should be called even class. In this paper we prefer to fit with the terminology given in the textbooks\cite{PayBook,Scott} even if the initial terminology given in \cite{KV1,KV2} and its natural extension would appear more natural to us.
	\end{rem}

	\begin{Proposition}
		$\phi$ is a projection and $\mathcal{F}Cl_{eo}(S^1,V) = Im \phi.$
	\end{Proposition}
	By the way, we also have
	$$ \mathcal{F}Cl(S^1,V) = \mathcal{F}Cl_{ee}(S^1,V) \oplus \mathcal{F}Cl_{eo}(S^1,V).$$
	We have the following composition rules for the class of a formal operator $A\circ B:$
	\vskip 12pt
	
	\begin{tabular}{|c|c|c|}
		\hline
		&& \\
		& $A$ even-even class & $A$ even class \\
		&& \\
		\hline
		&& \\
		$B$ even-even class & $A \circ B$ even-even class & $A \circ B$ even-odd class \\
		&& \\
		\hline
		&& \\
		$B$ even-odd class & $A \circ B$ even-odd class & $A \circ B$ even-even class \\
		&& \\
		\hline
	\end{tabular}
	
	\begin{example}
		$\epsilon(D)$ and $|D|$ are even-odd class, while we already mentioned that differential operators are even-even class.
	\end{example}

{
	\begin{rem}
		The operator $\epsilon(D)$ satisfies the following properties:
		\begin{itemize}
			\item Since $\epsilon(D)^2=Id,$ the left composition $A \in \F Cl(S^1,V)\mapsto \epsilon(D)\circ A$ is an involution on $\F Cl(S^1,V) $
			\item Since $\epsilon(D) \in \F Cl_{eo}(S^1,V),$ the  restriction of $\epsilon(D) \circ (.)$ to $\Psi DO(S^1,V)= \F Cl_{ee}(S^1,V)$ is a bijection from $\F Cl_{ee}(S^1,V)$ to $\F Cl_{eo}(S^1,V).$
		\end{itemize}
	\end{rem}
}
	One can also define the operator $s$ on $\F Cl(S^1,V)$ which extends the operator $s: T^*S^1 \rightarrow T^*S^1$ defined by $s(x,\xi) = (x,-\xi)$ by $$ s : \sum_n \sigma_{n}(x,\xi) \mapsto \sum_n  (-1)^n\sigma_{n}\left(s(x,\xi)\right).$$ 
This operator obviously satisfies $s^2 = Id,$ and we remark the following properties:
\begin{Proposition}
	\begin{itemize}
		\item $s\left( \F Cl_\pm(S^1,V)\right) = \F Cl_\mp(S^1,V)$
		\item $\F Cl_{ee}(S^1,V) = Ker(Id - s)$
		\item $\F Cl_{eo}(S^1,V) = Ker(Id + s)$
	\end{itemize}
\end{Proposition}
\begin{rem}
	One can consider also $s':\sum_n \sigma_{n}(x,\xi) \mapsto \sum_n \sigma_n \left(s(x,\xi)\right).$ We still have $s'^2 = Id,$  $s'\left( \F Cl_\pm(S^1,V)\right) = \F Cl_\mp(S^1,V)$ but the two other properties are not fulfilled.
\end{rem}
Under these properties, $\F Cl_{ee}(S^1,V)$ and $\F Cl_{eo}(S^1,V)$ appear respectively as eigen-spaces for the eigen values $1$ and $-1$ of the symmetry $s$, and hence an operator $a \in  \F Cl(S^1,V) = \F Cl_{ee}(S^1,V) \oplus \F Cl_{eo}(S^1,V)$ decomposes as $ a = a_{ee} + a_{eo}$
and $$ s[a,b] = [a,b]_{ee} - [a,b]_{eo} = \left([a_{ee},b_{ee}] + [a_{eo}, b_{eo}]\right) -\left([a_{ee},b_{eo}] + [a_{eo}, b_{ee}] \right).$$  

\subsection{Complex powers of a formal  pseudo-differential operator} \label{ss:complexpow}
Following \cite{FMR1993} inspired by \cite{See}, this is possible to define the complex power of an elliptic formal operator. Concerning formal operators,
ellipticity is fully obtained by a condition on the principal symbol of the operator. This provides the possibility, when the algebra of functions $R$ is e.g. a complete topological vector space with bounded addition and multiplication laws, 
to define complex powers $A^\alpha$ of a formal  operator $A$ for $\mathfrak{Re}(\alpha)<0$ via contour integrals similar to (\ref{eq:integral}) 
and then extend it to arbitrary complex powers. 
Beyond these technical problems, for any formal $\C-$algebra of functions $R$ with differentiation $\partial,$ it is possible to define the same complex powers of the Lax-type operators $L \in \Psi DO(R)$ present in the KP hierarchy, along the lines of \cite{KZ} and \cite{EKRRR1995}. Let $\alpha \in \C$ and let $\Psi DO^\alpha(R)$
be the affine space of formal series of the form 
$$ \sum_{k \in \N} a_{\alpha - k} \partial^{\alpha - k},$$ formally defined as $\Psi DO^\alpha(R) = \Psi DO^0(R). \partial^\alpha.$
On the total spce of formal pseudo-diferential operators of complex order
	 generated by the family  $\left(\Psi DO^\alpha(R)\right)_{\alpha \in \C},$ the same addition and multiplication rules as in $\Psi DO(R)$ holds true and consistent.}
Let $A  \in \Psi DO^\alpha(R)$ with  $a_\alpha \in \R_+^* \subset R,$  one can define $$\log(A) \in \alpha \log a\partial + \Psi DO^0(R)$$ such that $ \exp\left(\log(A)\right) = A$
by standard rules of formal series. 

Let $L \in \Psi DO^1(R)$ with principal symbol $\partial.$ We can then define the complex power $L^\alpha$ for $\alpha \in \C^*,$ and following the notations of \cite{EKRRR1995,KW}, the affine space
$\mathcal{L} =  \partial + \Psi DO^0(R)$
has an affine isomorphism, for $\alpha \in \C^*,$ with 
$$ \mathcal{L}^\alpha =  \partial^\alpha + \Psi DO^{\alpha-1}(R)$$
through the identification $ L \in \mathcal{L} \mapsto L^\alpha = \exp\left(\alpha\log(L)\right) \in \mathcal{L}^\alpha.$
From this construction on $\Psi DO(S^1,\K),$ one can push forward complex powers on subalgebras of $\F Cl(S^1,\K)$ via the identifications already described. More precisely, one use heuristically the bijection 
$\Phi_{1,0} : \Psi DO (S^1,\K) \rightarrow \F Cl_+(S^1,\K)$
to define, for $A = \frac{d}{dx}_+ + \sum_{k \leq 0} a_k \frac{d}{dx}^k_+\in \F Cl_+^1(S^1,\K),$
first the logarithm  $$\log A = \log \frac{d}{dx}_+ + \sum_{k \leq 0} a_k \frac{d}{dx}^k_+$$
and the complex power $A^\alpha = \exp\left(\alpha\log A\right)$
which formal symbol vanishes for $\xi <0.$
Then we define $$\F Cl^\alpha_+ (S^1,\K) = \F Cl^0_+ (S^1,\K)\frac{d}{dx}_+^\alpha$$ and (after these constructions) $\Phi_{1,0}$ extends naturally to a bijection from $\Psi DO^\alpha(S^1,\K)$ to $\F Cl^\alpha_+ (S^1,\K).$ The same construction holds to extend the identification of  $\Psi DO(S^1,\K)$ with $\F Cl_- (S^1,\K)$ to complex powers
$\Phi_{0,1}: \Psi DO^\alpha(S^1,\K) \rightarrow \F Cl_-^\alpha (S^1,\K)$
and define 
$$  \F Cl^\alpha (S^1,\K) =  \F Cl_+^\alpha (S^1,\K) \oplus  \F Cl_-^\alpha (S^1,\K) = (\Phi_{1,0}\times \Phi_{0,1})\left(\Psi DO^\alpha(S^1,\K)^2\right).$$
One can also understand $  \F Cl^\alpha (S^1,\K)$ as
$  \F Cl^\alpha (S^1,\K) =   \F Cl^0 (S^1,\K) |D|^\alpha$
where $|D|^\alpha = \Delta^{\frac{\alpha}{2}} $ is defined via Seeley's complex powers \cite{See}. Alternatively, setting $$\left(\frac{d}{dx}\right)^\alpha = \left(\frac{d}{dx}\right)_+^\alpha + \left(\frac{d}{dx}\right)_-^\alpha = i\epsilon(D) |D|^\alpha,$$
we get $  \F Cl^\alpha (S^1,\K) =   \F Cl^0 (S^1,\K) \left(\frac{d}{dx}\right)^\alpha.$
These spaces of complex powers contain the projections on formal operators (up to smoothing oprators) of the classes of pseudo-differential operators of complex order defined in 
\cite[section 3]{KV1}.


\subsection{Lie-algebraic digression} \label{ss:Lie}

\subsubsection{Operator bialgebras and Manin pairs}

One can easily define a Lie algebra structure by antysimmetrisation of the associative product $[A,B] = A\circ B - B\circ A$. We remark that the vector field Lie algebra ${\rm Vect}(S^1)$ and its semi-direct product with $C^{\infty}(S^1) = C^\infty(S^1,\C)$ is a natural Lie subalgebra of the differential operator Lie algebra $DO(S^1)$ which is formed by the order 1 differential operators and the order less or equal to 1. {This remark can be also deduced from Definition \ref{def:diff-op} by setting $E = S^1 \times \C,$ i.e. $V=\C.$ When $V = \C^n$ with $n \geq 2,$ an operator $X \in Vect(S^1)$ can be identified with the degree 1 differential opeartor $X \otimes Id_{\C^n} \in DO(S^1,V)$ while order $0$ differential operators coincide with multplication operators in $C^\infty(S^1,M_n(\C)).$ We also have that $C^\infty(S^1,M_n(\C)) \rtimes Vect(S^1) \subset DO^1(S^1,V)$ as a Lie algebra, but the off-diagonal operator $$A= \frac{d}{dx} \otimes \left(\begin{array}{cc}
	0 & 1 \\ 1 & 0
	\end{array}\right) = \left(\begin{array}{cc}
	0 & \frac{d}{dx} \\ \frac{d}{dx} & 0
	\end{array}\right) \in DO^1(S^1,V)$$
	is not an operator in the Lie algebra $C^\infty(S^1,M_n(\C)) \rtimes Vect(S^1).$ 
	One can always embed $DO(S^1)=DO(S^1,\C)$ into $DO(S^1,V)$ by identifying  $A \in DO(S^1)$ with $A \otimes Id_{\C^n} \in  DO(S^1,V).$ This identification is a morphism of unital algebras and a morphism of Lie algebras. 
}{It is a straightforward to check that the similar antisymmetrization of the product (\ref{asspsi}) gives a Lie algebra structure on $\Psi DO(S^1)$ and the algebra $DO(S^1)$ is a Lie subalgebra in it.}

One of the most exciting properties of this pair of infinite-dimensional Lie algebras is an existence of a trace functional (which is quite atypical in the infinite-dimensional world{). This functional is known as {\it Adler trace} $${\rm Tr}(A) = \oint_{S^1}tr_n( a_{-1} (x)) dx,$$ where $tr_n$ is the classical trace of $n \times n$ matrices,} and it defines a bilinear invariant symmetric form on $\Psi DO(S^1)$
$$ (A,B) \to {\rm Tr}(A\circ B),\quad A,B \in \Psi DO(S^1),$$
which is invariant with respect the multiplication: $(C\circ A,B) = (A,B\circ C)$ and also invariant with respect to the Lie bracket: $([C, A],B) = (A,[B, C])$ for any triple $A,B,C\in\Psi DO(S^1).$ 
{This form is a non-degenerate and  can be used to {build an injective map from} the algebra $\Psi DO(S^1)$ with its {\it dual}: to each $A\in \Psi DO(S^1)$ one can assign the linear functional $l_A \in (\Psi DO(S^1))^{\ast}$ such that $l_A (X) = {\rm Tr}(A\circ X)$}

{Let $A\in \Psi DO(S^1)$ such that it contains only {\it negative degrees} of the symbol $D =\partial$:
$$A = \sum_{k=-\infty} ^{-1} b_k (x) \partial^k.$$}
Such "purely Integral"operators are also closed with respect to both operations $\circ$ and $[,]$ and we shall denote this subalgebra in  $\Psi DO(S^1)$ by  $IO(S^1).$
It is easy to check that the subalgebra $DO(S^1)$ is dual to the subalgebra $IO(S^1)$ via the bilinear invariant form $(-,-)$ and the "full" algebra 
$\Psi DO(S^1) = DO(S^1) \oplus IO(S^1).$ Both subalgebras are isotropic  with respect to $(-,-).$ 
{The algebra triple $$(\Psi DO(S^1),DO(S^1), IO(S^1))$$ is known as a {\it Manin triple} and the algebra $DO(S^1)$ carries a structure of a {\it Lie bialgebra}.
We should admit that strictly speaking this triple and this bialgebra are not a genuine example of both structures in view of the following remark:}
{
\begin{rem}
We should remark that while  $DO(S^1) = (IO(S^1))^{\ast}$ the natural map  $IO(S^1)\to (DO(S^1))^{\ast}$ is not surjective since not every continuous linear functional on 
$C^{\infty}(S^1)$ is of the form $F \to (F,f), \quad F\in C^{\infty}(S^1)$ (\cite{Drin}).
\end{rem}
}
In what follows by abuse of the rigorous terminology ("pseudo-Manin triple", "pseudo-Lie bialgebra", "Khovanova triple" etc.) we shall call the operator triple above by
Manin triple and refer $DO(S^1)$ as a Lie bialgebra.

\subsubsection{Differential and integral part}

We first remind that if $V=\C^n$ and use the notations
$$ { DO}(S^1,V) = \bigcup_{o \in \N}\left\{ \sum_{0\leq k \leq o} a_k \partial^k \right\},\quad{ IO}(S^1,V) = \left\{ \sum_{k \leq -1} a_k \partial^k \right\}$$ 
we get also  the {vector space} decomposition 
\begin{eqnarray}
\label{psi-DS} \Psi DO(S^1,V) = { DO}(S^1,V)\oplus { IO}(S^1,V).
\end{eqnarray}
{such that any (matrix) order $k$ pseudo-differential operator $A = \sum_{i=-\infty}^{k}a_i \partial^i$ is splitted  in two components $A = A_{+} + A_{-}$ with 
$A_{+} =  \sum_{i=0}^{k}a_i \partial^i$ and $A_{-} =  \sum_{i=-\infty}^{-1}a_i \partial^i.$
In that case, when $V = \C^n$ and with obvious extension of notations, the algebra triple $\left(\Psi DO(S^1,V),DO(S^1,V), IO(S^1,V) \right)$ is known as a Manin triple and the algebra $DO(S^1,V)$ carries a structure of a  Lie bialgebra.}{
We shall use also (by abuse of notation)  the notation ${\rm Res}(A)$ for the {\it residue-matrix function}: 
$${\rm Res}: M_n(\Psi DO(S^1,\C)) \to C^{\infty}(S^1,M_n(\mathbb C)), \ , A \to a_{-1}(x)$$}
Let $A,B$ be some matrix-valued pseudo-differential operators, such that  $A = \sum_{i=-\infty}^{k}a_i \partial^i,\, B = \sum_{j=-\infty}^{l}b_j \partial^i$
with $a_j,b_j$ some matrix-valued functions. Then it is a straightforward exercise to check that  there exists a matrix-valued function $F$ such that
$${\rm Tr}([A,B]) = \oint tr_n ({\rm Res} [A,B])= \oint dF = 0.$$

\begin{rem}
The same holds when we replace concrete algebras of functions $C^\infty(S^1)$ by an abstract associative algebra $\mathcal R$ with unit element, equipped with integration properties, we refer to \cite{M1,M3} for a detailed description for the corresponding algebra of formal operators {$\Psi DO(	\mathcal R).$} Then, in presence of a non-trivial one-form 
$\oint : {\mathcal R} \rightarrow \C,$  one can define an analogous ot the Adler map that we note also  ${\rm Tr}$ by 
$$ {\rm Tr} : \sum_{k \in \Z} a_{k} \partial^k \mapsto \oint a_{-1}.$$ For example, when $\mathcal{R} = C^\infty(S^1,M_n(\C))$ for $n\geq 2,$ i.e. when 
$$\Psi DO(\mathcal R) = \Psi DO(S^1,\C^n)= M_n(\Psi DO(S^1,\C)),$$ 
the natural 1-form $\oint$ on $\mathcal R$ is exactly $\oint_{S^1} \circ tr_n$ already described. 
\end{rem}

\subsection{Poisson structures on matrix pseudo-differential operators}.\label{s:poisson}
{In analogy with the "scalar" ($n=1$) case one can define the first and the second Gelfand-Dikii Poisson structures in the framework of the formal Gelfand "variational" differential-geometric formalism in the infinite-dimensional setting. The results of this subsection are not new and are well-known since almost 30 years (see for example \cite{AdBilal}).}
{We define an infinite dimensional affine variety $L_k$ whose points, monic differential operators of order $k$, are defined by $k$ matrix function coefficients 
$\bar u=(u_1(x),\ldots, u_k(x))$  such that  $\forall j : 1\leq j\leq k, \, u_j (x) \in C^{\infty}(S^1,M_n(\mathbb C)):$
$$ L_k =\{ L = \partial^k + u_1\partial^{k-1} + \ldots + u_k \}.$$ }
{We consider a function algebra ${\mathcal C}(L_k)$ as a set of functionals $l : L_k \to \mathbb C$ of type
$$ l[\bar u] := \oint {\rm tr}( {\rm pol}( \partial^{\alpha}_x(u_j))), $$ where ${\rm pol}( \partial^{\alpha}_x(u_j))$ is a differential polynomial on $u_j (x).$}
{The tangent space to $L_k$ consists of differential operators of order $k-1$ and the cotangent space $T^{\ast} L_k$ can be identified with the quotient 
$IO (S^1,V)/IO_{-k}(S^1,V): $ via the coupling $T^{\ast} L_k \times T L_k \to {\mathcal C}(L_k),\quad  \langle X, V\rangle = {\rm Tr}(X\circ V).$ Here $X\in T^{\ast} L_k$ is the set of "covectors"  of the type $X = \sum_{j=1}^{k} \partial_x^{-j} \circ p_j, \, p_j \in{\rm pol}( \partial^{\alpha}_x(u_j).$}
{We shall remind the definition of {\it variational derivative} of a functional $l[\bar u] \in {\mathcal C}(L):$
$$
\frac{\delta l[\bar u]}{\delta u_j}(x)_{pq}=\sum_{s=o}^{\infty}(-1)^{s} \frac{d^{r}}{dx^{r}}\left( \frac{\partial {\rm tr}(\rm pol)(\bar u)(x)}{\partial(u_j^{(s)})_{pq}}\right), 1\leq p,q \leq n.
$$
The variational derivative assigns to each functional
$l[\bar u]\in \mathcal C(L)$ the pseudo-differential operator $$X_l = \sum_{r=0}^{k}\partial^{-r}\left(\frac{\delta l[\bar u]}{\delta u_{k+1-r}}\right).$$}
{Let $X_{l_{1,2}}$ be two such operators which can be interpreted as two covectors on $T^{\ast}L_k$. We define a family of brackets 
$$\{-,-\}_{\lambda} :\mathcal C(L)\times \mathcal C(L) \to \mathcal C(L) :$$
$$\{l_1, l_2\}_{\lambda}(L) = \oint tr_n ({\rm Res}((L+\lambda)(X_{l_1}(L+\lambda))_{+}X_{l_2} - ((L+\lambda)X_{l_1})_{+}(L+\lambda)X_{l_2}) =$$
$$\{l_1, l_2\}_2(L) +\lambda\{l_1, l_2\}_1(L)=$$
$$ \oint tr_n ({\rm Res}(L(X_{l_1}L)_{+}X_{l_2} - (LX_{l_1})_{+}LX_{l_2})) +\lambda \oint tr_n ({\rm Res}([L,X_{l_1}]_{+}X_{l_2}).$$}
{
\begin{Theorem}\label{AGD} (Adler-Gelfand-Dickey)
\begin{enumerate}
\item The family $\{-,-\}_{\lambda}$ is a family of Poisson structures on ${\mathcal C}(L)$;
\item The corresponding Hamiltonian map $ H_{\lambda} : T^{\ast} L_k \to T L_k :$ is given by
$$H_{\lambda}(X) = (LX)_{+}L  - L(XL)_{+}  +\lambda [L,X]_{+}, \quad X\in  T^{\ast} L_k, \, L\in L_k.$$
\item $H_{\lambda}(X) =H_{2}(X) +\lambda H_{1}(X)$ and each $V_i, i=1,2$ are Hamiltonian mappings.
\item The Hamiltonian maps $H_i$ relate to the Poisson brackets via 
$$ \{l_1, l_2\}_{\lambda}(L) = H_{\lambda}(\delta l_1)(l_2).$$
\item Covector fields $T^{\ast}L_k$ carry a Lie algebra structure with the bracket 
$$[X, Y] = [(XL)_{+}Y+(YL)_{-}X - X(LY)_{-} - Y(LX)_{+} + H_2(X)Y)-H_2(Y)(X)]_{-}$$
which will be called {\rm the second Gelfand-Dikii algebra} $GD_2.$
\end{enumerate}
\end{Theorem}
This structure relates in some sense to the Manin triple on $\Psi DO(S^1,V).$ }

{
\subsubsection{Semenov-Tyan-Shansky $r-$matrix  construction}
Let $A_{\pm}$ two elements of the Lie algebra $\Psi DO(S^1,V)$ such that $A_{+} \in DO(S^1,V)$ and $A_{-}\in IO(S^1,V).$ Then one can identify $\Psi DO(S^1,V)\otimes \Psi DO(S^1,V)$ with ${\rm Hom} (\Psi DO(S^1,V),\Psi DO(S^1,V)$ using the inner product on $\Psi DO(S^1,V).$ Therefore, if  we consider the bi-vector ${\bold r} \in \Lambda^{2} (\Psi DO(S^1,V))$ such that
$\langle {\bold r}, A^{\ast}_{+}\wedge A^{\ast}_{-}\rangle = (A_{+}, A_{-}) = {\rm Tr}(A_{+}\circ A_{-}),$
where $A^{\ast}$ is a dual to $A$ with respect to the inner product., then we can identify it with the operator ${\tilde{\bold r}}\in {\rm End}(\Psi DO(S^1,V))$ such that
${\tilde{\bold r}}\vert_{DO(S^1,V)} = 1,\quad {\tilde{\bold r}}\vert_{IO(S^1,V)} = -1.$}
\subsubsection{Analogues of splittings}
Back to $\F Cl(S^1,V),$ the maps $$ A \in \F Cl(S^1,V) \mapsto \sum_{k \in \Z} \sigma_{k}(A)(x,1) \partial^k$$
and  $$ A \in \F Cl(S^1,V) \mapsto \sum_{k \in \Z} \sigma_{k}(A)(x,-1) \partial^k,$$ identify $\Psi DO(S^1,V)$ with $\F Cl_+(S^1,V)$ for the first one and $\F Cl_-(S^1,V)$ for the second one. 

Thus, there exists a decomposition $\F Cl_+(S^1,V) = \F Cl_{+,D}(S^1,V) \oplus\F Cl_{+,S}(S^1,V) $ and another 
$\F Cl_-(S^1,V) = \F Cl_{-,D}(S^1,V) \oplus\F Cl_{-,S}(S^1,V), $ and setting $$ \F Cl_D(S^1,V) = \F Cl_{+,D}(S^1,V) \oplus\F Cl_{-,D}(S^1,V),  $$
$$ \F Cl_S(S^1,V) = \F Cl_{+,S}(S^1,V) \oplus\F Cl_{-,S}(S^1,V),  $$
we get the vector space decomposition analogous to (\ref{psi-DS}):
$$ \F Cl(S^1,V) = \F Cl_{D}(S^1,V) \oplus\F Cl_{S}(S^1,V),  $$
{
\subsection{Manin pairs on $\F Cl(S^1,V)$} \label{s:Manin}
\subsubsection{Extension of the classical Manin triple to $\F Cl(S^1,V)$} \label{ss:manin}
The Adler trace \cite{Adl} defined by $$ Tr: A = \sum_{k \leq o} a_k \partial^k \mapsto \int_{S^1} tr(a_{-1})$$ is the only non trivial trace on $ \Psi DO(S^1,V).$ Morover, see e.g. \cite{EKRRR1995} and \cite{KZ},
\begin{Theorem}
	$(\Psi DO(S^1,V), { IO}(S^1,V), { DO}(S^1,V), Tr )$ is a Manin triple.
\end{Theorem}
The Wodzicki residue (\cite{Wod1984}, see e.g. \cite{Ka}) is usually known as an ``extension'' of the Adler trace to $\F Cl(S^1,V)$ and hence to $Cl(S^1,V).$ For the sake of deeper insight on what is described in the rest of this paper, we need to precise that the space of traces on $\F Cl(S^1,V)$ is 2-dimensional, generated by two functionals: $$ res_+: A \mapsto \int_{S^1} \sigma_{-1}(A)(x,1) |dx|$$
and  $$ res_-: A \mapsto \int_{S^1} tr( \sigma_{-1}(A))(x,-1) |dx|.$$
The functionals $res_\pm$ are the only non-vanishing traces on $\F Cl_\pm(S^1,V)$ (up to a scalar factor) and are vanishing on $\F Cl_\mp(S^1,V).$ The (classical) Wodzicki residue reads as $res = res_+ + res_-.$
Because the partial symbol $\sigma_{-1}(A)$ of an operator $A \in \Psi DO(S^1,V)$ is skew-symmetric in the $\xi-$variable, $res$ is vanishing on $\Psi DO(S^1,V) = \F Cl_{ee}(S^1,V), $ so that it is superficial to state that the Wodzicki residue is ``simply'' the extension of the Adler trace.  However the two linear functionals already described, namely 
$$ A \in \F Cl(S^1,V) \mapsto \sum_{k \in \Z} \sigma_{k}(A)(x,1) \partial^k$$
and  $$ A \in \F Cl(S^1,V) \mapsto \sum_{k \in \Z} \sigma_{k}(A)(x,-1) \partial^k,$$ identity $res_+$ and $res_-$ respectively with $Tr.$ By the way, we can state:
\begin{Theorem} We have three Manin triples: 
	$$(\F Cl_+(S^1,V), \F Cl_{+,S}(S^1,V), \F Cl_{+,D}(S^1,V), res_+ ), $$ 
	$$(\F Cl_-(S^1,V), \F Cl_{-,S}(S^1,V), \F Cl_{-,D}(S^1,V), res_- )$$ and $$(\F Cl(S^1,V), \F Cl_{S}(S^1,V), \F Cl_{D}(S^1,V), res ).$$
\end{Theorem}
\subsubsection{A remark on two "non-invariant Manin triples"}
{ Following \cite{EKRRR1995}, given an operator $\mathbf r$ acting on $\F Cl(S^1,V)$ satisfying ${\bold r}^2=Id,$
	one can form a $\F Cl(S^1,V)-$valued skew-symmetric bilinear form 
	$$ [.,.]_{\mathbf  r} = \frac{1}{2}\left([{\mathbf r}(.),.] + [.,{\mathbf  r}(.)]\right).$$
	In what follows, we concentrate on the cases ${\mathbf  r} = \epsilon(D)\circ (.),$ ${\mathbf  r} = s$ and also ${\mathbf  r} = s'.$ The corresponding brackets will be noted respectively by $[.,.]_{\epsilon(D)}$, $[.,.]_s$ and $[.,.]_{s'}.$}
	
Let us define $(A,B)_{s'} = res(A,s'(B)).$ By direct calculations, we find successively:
}
{
\begin{Lemma}\label{lemma:s'-pair}
	$(.;.)_{s'}$ is non degenerate and symmetric.
\end{Lemma}
\begin{Theorem}
	On $\F Cl(S^1,V) = \F Cl_+ (S^1,V) + \F Cl_-(S^1,V),$ $(.;.)_{s'}$ is a non degenerate and symmetric bilinear from for which the Lie algebras $\F Cl_+ (S^1,V)$ and $\F Cl_- (S^1,V)$ are isotropic.
\end{Theorem}
Let us define $(A,B)_s = res(A,s(B)).$
\begin{Lemma}\label{lemma:s-pair}
	$(.;.)_{s}$ is non degenerate and skew-symmetric but neither invariant for $[.,.]$ nor for $[.,.]_{\epsilon(D)}.$
\end{Lemma}
\begin{Theorem}
	On $\F Cl(S^1,V) = \F Cl_+ (S^1,V) + \F Cl_-(S^1,V),$ $(.;.)_{s}$ is a non degenerate and skew-symmetric bilinear from for which the Lie algebras $\F Cl_+ (S^1,V)$ and $\F Cl_- (S^1,V)$ are isotropic.
\end{Theorem}
\subsubsection{Two other  Manin pairs}
Let us consider the decomposition $$\F Cl(S^1,V) = \F Cl_{ee} (S^1,V) + \F Cl_{eo}(S^1,V),$$
that we equip with the classical Lie bracket $[.,.]$ or with $[.,.]_{\epsilon(D)}.$ and with the bilinear form $(A,B)=res(AB).$
\begin{Theorem} \label{1.25}
	res(AB) is a bilinear, non degenerate, symmetric and invariant form for both brackets, and $\F Cl_{ee} (S^1,V)$ as well as $\F Cl_{eo}(S^1,V)$ are isotropic vector spaces. Moreover, 
	\begin{itemize}
		\item for $[.,.],$ $\F Cl_{ee} (S^1,V)$ is a Lie algebra
		\item for $[.,.]_{\epsilon(D)},$  $\F Cl_{eo}(S^1,V)$ is a Lie algebra.
	\end{itemize}
\end{Theorem}
}
\subsection{Polarized Lie bracket} \label{ss:polar}

{
The } modified Yang–Baxter equation gives the condition on $\bold r$ for making $ [.,.]_{\bold r}$ a Lie bracket:
$$  [\mathbf  r X , \mathbf r Y ] - \mathbf  r ([\mathbf r X , Y ] + [X , \mathbf r Y ]) = - [X , Y ].$$
By direct computations, we get the following:
\begin{Theorem}\label{th:brackets}
	On the vector space $\F Cl(S^1,V),$
	\begin{enumerate}
	\item
	$[.,.]_{\epsilon(D)}$ is a Lie bracket  for which $[\F Cl_{ee}(S^1,V),\F Cl_{ee}(S^1,V)]_{\epsilon(D)} \subset \F Cl_{eo}(S^1,V)$ and
	$[\F Cl_{eo}(S^1,V),\F Cl_{eo}(S^1,V)]_{\epsilon(D)} \subset \F Cl_{eo}(S^1,V).$
\item
	$[.,.]_{s}$ and $[.,.]_{s'}$ are not Lie brackets. \end{enumerate}
\end{Theorem}

\begin{rem}\textbf{(Testing Rota-Baxter equations and Reynolds operators)}
Testing by direct calculations
the Rota-Baxter equations
$$ R(u)R(v) - R(R(u)v) - R(uR(v)) = \lambda R(uv)$$
 for a weight $\lambda \in \C,$ one finds that $R = \epsilon(D) \circ (.),$ $R=s$ and $R=s'$ do not satisfy the Rota-Baxter equations (i.e. don't define new associative algebra operations) The same calculations show that these are not Reynolds operators (i.e. they do not satisfy the condition
$ R(R(u)v) = R(u)R(v)$ for all $u, v$ in the underlying associative algebra).  \end{rem}

	\subsection{Preliminaries on the KP hierarchy} \label{ss:prelKP}
	Let  $R$ be an algebra of functions equipped with a derivation $\partial.$ For us, $R = C^\infty(S^1,\K)$ with $\K = \R, \C$ and  $\mathbb{H},$ and $\partial = \frac{d}{dx}.$ In this context, where algebras of functions $R$ are Fr\'echet algebras, a natural notion of differentiability occurs, making addition, multiplication and differentiation smooth. By the way, considering addition and multiplication in $\Psi DO(S^1,\K),$ one can say that addition and multiplication in $\Psi DO(S^1,\K)$ by understanding, under this terminology, that, if $A = \sum_{n \in \Z} a_n \partial^n $ and $ B = \sum_{n \in \Z} b_n \partial^n ,$
	setting $A+ B = C = \sum_{n \in \Z} c_n \partial^n$
	and $AB = D = \sum_{n \in \Z} d_n \partial^n$
	the map $$ \left((a_n)_{n \in \Z},(b_n)_{n \in \Z}\right) \mapsto \left((c_n)_{n \in \Z},(d_n)_{n \in \Z}\right)$$ is smooth in the relevant infinite product. We make these precisons in other to circumvent the technical tools recently developed in \cite{ERMR,MR2016} where a fully rigorous framework for smoothness on these objects is described and used.
		Let $T=\{t_n\}_{n \in \N^*}$ be an infinite set of formal (time) variables and let us consider the 
	algebra of formal series 
	$\Psi DO(S^1,\K)[[T]]$ with infinite set of formal variables $t_1,t_2,\cdot$ with $T-$valuation $val$ defined by $val_T(t_n) = n$ \cite{M1}. One can extend naturally on $\Psi DO(S^1,\K)[[T]]$ the notion of smoothness from the same notion on $\Psi DO(S^1,\K),$ see \cite{MR2016} for a more complete description.    
	The Kadomtsev-Petviashvili (KP) hierarchy reads
	\begin{equation} \label{eq:KP}
	\frac{d L}{d t_{k}} = \left[ (L^{k})_{D} , L \right]\; , \quad
	\quad k \geq 1 \; ,
	\end{equation}
	with initial condition $L(0)=L_0  \in \partial + \Psi^{-1} (R)$. The dependent
	variable $L$ is chosen to be of the form
	$L = \partial + \sum_{\alpha \leq -1 } u_\alpha \partial^\alpha \in {\Psi}^1(S^1,\K)[[T]] \; .$
	A standard reference on (\ref{eq:KP}) is L.A. Dickey's treatise \cite{D}, see also 
	\cite{KZ,M1,M3}. 
	In order to solve the KP hierarchy, we need the following groups (see e.g. \cite{MR2016} for a latest adaptation of Mulase's construction \cite{M1,M3}):
	$$ \bar{G} = 1 + \Psi DO^{-1}(S^1,\K)[[T]],$$
	$$ \overline{\Psi} = \left\{ P = \sum_{\alpha \in {\mathbb{Z}}}
	a_{\alpha}\,\partial^{\alpha} \in {\Psi}(S^1,\K)[[T]] \; : 
	\, val_T(a_\alpha)\geq \alpha \hbox{ and } P|_{t = 0} \in 1 + \Psi DO^{-1}(S^1,\K)  \right\}$$ 
	and 
	$$\overline{\D} = \left\{ P= \sum_{\alpha \in \mathbb{Z}}
	a_{\alpha}\,\partial^{\alpha} : P \in\overline{\Psi}(A_t) \mbox{ and }
	a_\alpha=0 \mbox { for } \alpha <0 \right\} \; .$$
	We have a matched pair $ \overline{\Psi} = \bar{G} \bowtie \overline{\D}$ which is smooth under the terminology we gave before.
	The following result, from \cite{MR2016}, gives a synthesied statement of main results on the KP hierarchy (\ref{eq:KP}) and states smooth dependence on the initial conditions in the case where $R$ is commutative (i.e. $R = C^\infty(S^1,\R)$ or $R = C^\infty(S^1,\C)$ in this work).
	
	\begin{Theorem} \label{KPcentral} 
		\cite{MR2016}
			Consider the KP hierarchy 
		\ref{eq:KP}
			with initial condition $L(0)=L_0$. Then,
			\begin{enumerate}
				\item There exists a pair $(S,Y) \in \bar{G} \times \overline{\D}$ 
				such that the unique solution to Equation $(\ref{eq:KP})$ with $L|_{t=0}=L_0$ is
				$
					L(t_1,t_2,\cdots)=Y\,L_0\,Y^{-1} = S L_0 S^{-1}$.
				\item The pair $(S,Y)$ is uniquely determined by the smooth decomposition problem 
				$$exp\left(\sum_{k \in \N}\tau_k L_0^k\right) = S^{-1}Y$$ 
				and the solution $L$ depends smoothly on the initial condition $L_0$.
			
				\item The solution operator $L$ is smoothly dependent  on the 
				initial value $L_0.$ 
			\end{enumerate}
	\end{Theorem}
	We now describe the case $\K = \mbH = \R + i\R + j\R + k\R.$ The algebra $\Psi DO(S^1,\mbH)$ is constructed from the non commutative Fr\'echet algebra
	$$ C^\infty(S^1,\mbH) =  C^\infty(S^1,\R) \oplus iC^\infty(S^1,\R) \oplus jC^\infty(S^1,\R) \oplus kC^\infty(S^1,\R).$$ All the constructions before remain valid following \cite{Ku2000,McI2011}, setting $V= \mbH$ as a 4-dimensional $\R-$algebra, and the algebraic description of the solutions of the KP hierarchy (\ref{eq:KP}) with $L_0 \in \Psi DO^1(S^1,\mbH)$ and $L \in \Psi^1 DO(S^1,\mbH)[[T]]$ as before can be completed by stating that the coefficients of the $T-$series of $L$ depend smoothly on the initial value $L_0$ from \cite{ERMR}. 
	\section{Injecting $\Psi DO$ into $\mathcal{F} Cl.$}
	
	\subsection{Injecting $\Psi DO(S^1,\K)$ in $\mathcal{F} Cl(S^1,\K).$} \label{s:inj}
	We already mentionned the identification of 
	$\Psi DO (S^1,\K)$ with $\F Cl_{ee}(S^1,\K),$ present when $\K = \R$ or $\C$ in \cite{MR2018}. We claim here that this identification also applies straightway when $\K = \mbH.$ We denote by $\Phi_{ee}$ this identification, that can be generalized to $$\Phi_{ee,\lambda} : \sum_{k \in \Z} a_{k} \left(\frac{d}{dx}\right)^k \in \Psi DO(S^1,\K) \mapsto \sum_{k \in \Z} a_{k} \left(\lambda\frac{d}{dx}\right)^k \in \F Cl_{ee}(S^1,\K) . $$
	Similar to this identification, we have other { injections} { for } $\lambda \in \R^*:$ 
	$$\Phi_{\epsilon(D),\lambda} : \sum_{k \in \Z} a_{k} \left(\frac{d}{dx}\right)^k \in \Psi DO(S^1,\K) \mapsto \sum_{k \in \Z} a_{k} \left(\lambda\epsilon(D)\frac{d}{dx}\right)^k \in \F Cl(S^1,\K) , \hbox{ and}$$ 
	$$\Phi_{\lambda,\mu} : \sum_{k \in \Z} a_{k} \left(\frac{d}{dx}\right)^k \in \Psi DO(S^1,\K) \mapsto \sum_{k \in \Z} a_{k} \left(\lambda^k\left(\frac{d}{dx}\right)_+^k + \mu^k\left(\frac{d}{dx}\right)_-^k\right)  \in \F Cl(S^1,\K) $$
	for {$(\lambda,\mu) \in \C^2 \backslash \{(0;0)\},$} with unusual convention $0^k = 0$ $\forall k \in \Z.$ 
	\begin{rem}
		$\Phi_{1,1} = \Phi_{ee}$ and $\Phi_{1,-1} = \Phi_{\epsilon(D),1}.$
	\end{rem}
	
	\begin{rem}
		$Im \Phi_{1,0} = \F Cl_+(S^1,\K) $ and $\Phi_{1,0}$ is a isomorphism of algebras from $\Psi DO(S^1,\K)$ to $\F Cl_+(S^1,\K) .$ The same way, $\Phi_{0,1}$ identifies the algebras $\Psi DO(S^1,\K)$ and $\F Cl_-(S^1,\K) .$
	\end{rem}
{
\begin{rem}
	Wa have also to say that the maps $\Phi_{\lambda,\mu}$ are not algebra morphisms unless $(\lambda,\mu) \in \{(1;0),(0;1),(1;1)\}.$ For example, let $\lambda \in \mathbb{C}-\{0;1\}.$ the map $\Phi_{\lambda,0}$ pushes forward the multiplication on $\Psi DO(S^1,\K)$ to a deformed composition $*_k$ on $\F Cl_+(S^1,\K)$ that reads as
	$\sigma(A) *_k \sigma(B) = \sum_{\alpha \in \N} \frac{(-i)^\alpha}{ \alpha!.k^\alpha} D^\alpha_x\sigma(A) D^\alpha_\xi\sigma(B).$
\end{rem} }
	Let us now give some sample images:
	\vskip 12pt
	\begin{tabular}{|c|c|c|c|c|}
		\hline  &&&& \\
		$A \in \Psi DO(S^1,\C)$ &$1$ & $\frac{d}{dx}$& $-\frac{d}{dx} = \Delta$ & $\left(1 +\frac{d}{dx} \right)^{-1}$ \\ &&&& \\
		\hline &&&& \\
		$\Phi_{\epsilon(D),1}(A)$ &$1$&$\epsilon(D)\frac{d}{dx} = i|D|$&$\Delta$& $\left(1 +\epsilon(D)\frac{d}{dx} \right)^{-1}$ \\ &&&& \\
		\hline &&&& \\ $\Phi_{ee,-1}(A)$ &$1$&$-\frac{d}{dx}$&$\Delta$&$\left(1 -\frac{d}{dx} \right)^{-1}$\\ &&&& \\
		\hline &&&& \\
		$\Phi_{1,0}(A)$ &$1_+$&$\left(\frac{d}{dx}\right)_+$&$\Delta_+$&$\left(\left(1 +\frac{d}{dx} \right)^{-1}\right)_+$\\ &&&& \\
		\hline
	\end{tabular}
\vskip 12pt
From our previous remarks, we get:
\begin{Theorem}
	The map $$\Phi_{1,0} \times \Phi_{0,1} : \Psi DO(S^1,\K)^2 \rightarrow \F Cl_+(S^1,\K) \times \F Cl_-(S^1,\K) = \F Cl(S^1,\K)$$ is an isomorphism of algebra.
\end{Theorem}
We also remark a new subalgebra of $\F Cl(S^1,\K):$
\begin{Definition}
	Let $\F Cl_\epsilon(S^1,\K)$ be the image of $\Phi_{\epsilon(D),1}$ in $\F Cl(S^1,\K).$
\end{Definition}
We have the obvious identification $\F Cl_\epsilon(S^1,\K) = C^\infty(S^1,\K)((i|D|^{-1}))$ as a vector space. 
	\subsection{Identification of $\F Cl(S^1,\C)$ with $\Psi DO(S^1,\mbH).$ }\label{s:id}	Let $$ i \epsilon(D) = \left( \frac{d}{dx}\right).|D|^{-1} = |D|^{-1}.\left( \frac{d}{dx}\right) .$$ We define the operator $J_1 = i \epsilon(D) \circ (.)$ on  $\mathcal{F}Cl(S^1,V).$
	\begin{Theorem}\label{th:J1}
		The operator $J_1$ defines an integrable almost complex structure on $\mathcal{F}Cl(S^1,V).$$\F Cl(S^1,V) = \Psi DO(S^1,V)\otimes\C$ as a real algebra, identifying $\F Cl_ee(S^1,V)$ with $\Psi DO(S^1,V)$ (real part) and  $\F Cl_eo(S^1,V)$ with $i\Psi DO(S^1,V)$ (imaginary part).
	\end{Theorem}
	
	We now identify two other almost complex structures:
	$J_2 = i s(.),$ $ J_3 = is'$ and 
	Clearly, $\forall i \in \{2;3\}, J_i^2 = -Id$ and we have also:
	\begin{Proposition}\label{prop:J1J2}
		$J_1 \circ J_2 = - J_2 \circ J_1$	
	\end{Proposition} 
	\begin{Theorem} \label{th:J2int}
		The operator $J_2$ defines a non integrable  almost complex structure on $\mathcal{F}Cl(S^1,V).$ Hence, gathering all these results, we get that 	the almost quaternionic structure$(J_1,J_2)$ is non integrable.
	\end{Theorem}

\begin{Proposition}\label{prop:J1J3}
$J_1 \circ J_3 = - J_3 \circ J_1$	
\end{Proposition}
\begin{Proposition}\label{prop:J3J2}
	$J_3 \circ J_2 = J_2J_3 \neq - J_2 \circ J_3$	
\end{Proposition}
\begin{Theorem} \label{th:J3int}
	The operator $J_3$ defines a non integrable  almost complex structure on $\mathcal{F}Cl(S^1,V).$
		The almost quaternionic structure $(J_1,J_3)$ is non integrable.
	\end{Theorem}
Let us now define $J_4 = J_1J_3.$
\begin{Proposition}\label{Prop:J4}
	We have:
	\begin{itemize}
		\item $J_4^2 = - Id.$
		\item $J_2J_4 = - J_4J_2.$
		\item $J_1J_4 = - J_4J_1.$
	\end{itemize}
\end{Proposition}

	\section{KP hierarchy with integer and complex order Lax operators in  $\F Cl(S^1,\C)$ and $\Psi DO(S^1,\mbH).$}\label{s:KP}

	\subsection{Multiple classical KP hierarchies on $\mathcal{F}Cl(S^1,\K)$} \label{ss:multKP}
	
	The (classical) KP hierarchy on $\Psi DO(S^1,\K)$ can then push-forward on $\mathcal{F}Cl$-classes of operators by various ways: 
	\begin{itemize}
		\item via identifications of $\Psi DO(S^1,\K)$ with subalgebras or ideals of $\mathcal{F}Cl(S^1,\K),$ for $\K = \R, \C$ or $\mathbb{H}.$
		\item by changing the standard multiplication of $\mathcal{F}Cl(S^1,\K)$ for $\K = \R, \C$ or $\mathbb{H},$ by ``twisting it'' by the operator $\epsilon(D)$ or $i\epsilon(D).$
		\item via the almost quaternionic structures that we identified on $\mathcal{F}Cl(S^1,\C)$ in order to identify it with $\Psi DO(S^1,\mathbb{H})$	 
	\end{itemize} 
	Let us describe in a detailed way these different approaches. 
	\paragraph{\bf Push-Forward via $\Phi_{\lambda,\mu}$ maps}
	Let $\K =  \C$ or $\mathbb{H}.$ For each choice of $(\lambda,\mu) \in \C^2\backslash \{0;0\}$ identifies $\frac{d}{dx} \in \Psi DO(S^1,\K)$ with an operator in $\mathcal{F}Cl(S^1,\K)$ with the same algebraic properties. 
	\vskip 6pt
	\noindent
	\textbf{Notation:} $\partial_{\lambda,\mu} = \Phi_{\lambda,\mu}\left(\frac{d}{dx}\right)$ and $\F Cl_{\lambda,\mu}(S^1,\K) = Im \Phi_{\lambda,\mu}.$
	\vskip 6pt
	Then we can develop the KP hierarchy on $\F Cl_{\lambda,\mu}(S^1,\K).$ We first remark that, since each map $\Phi_{\lambda,\mu}$ is a degree $0$ morphism of filtered algebras, each push-forward of the unique solotion $L$ of the KP hierachy (\ref{eq:KP}) generates a solution of the corresponding equation in $\F Cl(S^1,\K)$ which reads the same way: 
	$$\frac{d L}{d t_{k}} = \left[ (L^{k})_{D} , L \right]\; , \quad
	\quad k \geq 1 \; ,$$
	where solutions operators now belong to $\F Cl^1(S^1,\K)[[T]]$ and where each initial value $\Phi_{\lambda,\mu}(L_0) \in \partial_{\lambda,\mu} + \F Cl_{\lambda,\mu}^{-1}(S^1,\K)$ with obvious extension of notations.
	 Therefore, for any initial value $L_0 \in \Psi DO(S^1,\K),$ we get a family of operators $$L_{\lambda,\mu} \in\F Cl_{\lambda,\mu}^1(S^1,\K)[[T]] \subset \F Cl^1(S^1,\K)[[T]]$$ parametrized by the complex parameters $\lambda$ and $\mu$ chosen as before, which satisfies the KP hierarchy in $\F Cl(S^1,\K)$ and with initial values $\Phi_{\lambda,\mu}(L_0).$ 
	 \paragraph{\bf Existence, uniqueness and well-posedness of the KP system in $\F Cl(S^1,\K).$}
	 We adapt here the $r-$matrix approach for the construction of the solutions, along the lines of \cite{ERMR} with the following specific choices:
	 \begin{itemize}
	 	\item The algebra of smooth coefficients for formal pseudo-differential operators is $  \mathcal{R} = C^\infty(S^1,M_n(\K)) \oplus \epsilon(D)C^\infty(S^1,M_n(\K))$ with multiplication rules inherited from $Cl(S^1,\K^n).$
	 	\item The differential operator is $\partial = \frac{d}{dx}.$
	 \end{itemize}
 \begin{Proposition}\label{prop:R-Kn}
 	$\Psi DO(\mathcal{R}) = \F Cl(S^1,\K^n)$ and there is an identification of the Manin triples $(\Psi DO(\mathcal{R}), DO(\mathcal{R}), IO(\mathcal{R}))$ with $(\F Cl(S^1,\K^n),\F Cl_D(S^1,\K^n),\F Cl_S(S^1,\K^n)).$
 \end{Proposition}

	 Hence, applying the main result of \cite{RS1981} completed, for well-posedness, by \cite[Theorem 4.1]{ERMR} or by \cite[Theorem 4.1]{MR2016} when $\mathcal R = C^\infty(S^1,\K) = M_1(C^\infty(S^1,\K))$ is a commutative algebra, we can state the following: 
	 	
	 	\begin{Proposition} \label{prop:KP-FCl-I}
	 	The Kadomtsev-Petviashvili (KP) hierarchy (\ref{eq:KP}) on $\Psi DO (\mathcal R)$ (resp. $\F Cl(S^1,\K^n)$) 
	 with initial condition $L(0)=L_0  \in \partial + \Psi DO^{-1} (\mathcal R)$ (resp. $\in \partial + \F Cl^{-1}(S^1,\K^n)$) satisfies Theorem \ref{KPcentral}.
	\end{Proposition}
\begin{rem}
	We have used here, intrinsically, the integrable almost complex structure $J_1.$ Indeed, $\mathcal R = C^\infty(S^1,M_n(\K)) + J_1 C^\infty(S^1,M_n(\K))$ is an algebra.
\end{rem}
\begin{rem}
	There exists another way to justify Proposition \ref{prop:KP-FCl-I}. One can use alternatively the splitting $$\F Cl(S^1,\K^n) = \F Cl_+(S^1,\K^n) \oplus \F Cl_-(S^1,\K^n).$$ Then Equation (\ref{eq:KP}) on $\F Cl(S^1,\K^n)$ splits into two independent equations, similar to Equation (\ref{eq:KP}) on $\F Cl_\pm(S^1,\K^n).$ Through the identification maps $\Phi_{1,0}$ and $\Phi_{0,1}$ of $\F Cl_\pm(S^1,\K^n)$ with $\Psi DO(S^1,\K^n),$ we get existence, uniqueness and well-posedness for  Equation (\ref{eq:KP}) on $\F Cl(S^1,\K^n)$ with initaial value $L_0 \in \partial + \F Cl^{-1}(S^1,\K^n).$ 
\end{rem}

From this last remark, we can generalize the  identification procedure, changing the maps $\Phi_{ee},$ $\Phi_{1,0}$ and $\Phi_{0,1}$ by the family of maps $\Phi_{\lambda,\mu}.$

	 
	 \begin{Theorem} \label{th:KPlambdamu}
	 	Let $(\lambda,\mu) \in (\C^*)^2 .$ Then the KP equation (\ref{eq:KP}) in $\F Cl(S^1,\K^n)$ with initial value $L_0 \in \partial_{\lambda,\mu} + \F Cl^{-1}(S^1,\K^n)$ has an unique solution $L$ in $\partial_{\lambda,\mu} + \F Cl^{-1}(S^1,\K^n)[[T]]$ and the problem is well-posed: the solution $L$ depends smoothly on $L_0.$ 
	 \end{Theorem}
	 \paragraph{\bf Twisted KP hierarchy}
	 Let us now change the standard multiplication on $\F Cl(S^1,\K)$ by 
	 $ (A,B) \mapsto \epsilon A B$
	 where $ \epsilon = \epsilon(D)$ or $a\epsilon(D)$ for any $a \in \C^*.$ Since $\epsilon(D)$ commutes with any element of $\F Cl(S^1,\K)$ for the  standard multiplication, this new multiplication defines a new algebra structure on $\F Cl(S^1,\K).$ When necessary we note by $\circ$ the standard multiplication, and by $\circ_\epsilon$ the twisted one.  Associated to this multiplication, we get the deformed Lie bracket $[.,.]_\epsilon.$
	 Then we get again and equation similar to (\ref{eq:KP})
	 \begin{equation} \label{eq:KPepsilon}
	 \frac{d L}{d t_{k}} = \left[ \epsilon^{k-1}(L^{k})_{D} , L \right]_\epsilon =\epsilon^{k}\left[ (L^{k})_{D} , L \right] \; , \quad
	 \quad k \geq 1 \; ,
	 \end{equation}
	 where powers in this equation are taken with respect to $\circ.$
	 
	 \begin{Theorem} \label{th:KPepsilon}
	 	The Let $L_0$ such that $L_0 \in \partial_{\lambda,\mu} + \F Cl^{-1}(S^1,\K^n),$ with $(\lambda,\mu)\in (\C^*)^2.$ Then the $\epsilon-$KP hierarchy (\ref{eq:KPepsilon}) with initial value $L_0$ has an unique solution. Moreover, the problem is well-posed.
	 \end{Theorem}

	
	\subsection{KP hierarchies with complex powers} \label{ss:complex}
	We finally extend all the constructions of the last section to complex powers, along the lines of \cite{EKRRR1995}. Let $\K = \C $ or $\mathbb{H}.$
	We consider an operator $L_0$ of complex order $\alpha$ such that 
	\begin{equation} \label{eq:complex1} L_0 \in  \left(\frac{d}{dx}\right)^\alpha + \F Cl^{\alpha-1}(S^1,\K) \end{equation}
	or
	 \begin{equation} \label{eq:complex2} L_0 \in  \left|D\right|^\alpha + \F Cl^{\alpha-1}(S^1,\K) \end{equation}
	 For each setting (\ref{eq:complex1}) and (\ref{eq:complex2}), we define the complex KP hierarchy on $\F Cl^\alpha(S^1,\K)$ by \begin{equation} \label{eq:KPalpha}
	 \frac{d L}{d t_{k}} = \left[ (L^{k/\alpha})_{D} , L \right]_\epsilon =-\left[ (L^{k})_{S} , L \right] \; , \quad
	 \quad k \geq 1 \; ,
	 \end{equation}
	 where $L^{k/\alpha} = \exp\left(\frac{k}{\alpha}\log L\right)$
	 and the solution $L \in \F Cl^\alpha(S^1,\K)[[T]].$
	 
	 \begin{Theorem} \label{th:KPalpha}
	 	The KP hierarchy (\ref{eq:KPalpha}) with initial value $L_0$ defined along the lines of (\ref{eq:complex1}) or (\ref{eq:complex2}) has an unique solution in $\F Cl^\alpha(S^1,\K)[[T]].$ Moreover, the prblem is well-posed.
	 \end{Theorem}
 
 \section{Hamiltonian approaches}
 Now we consider $\F Cl(S^1,\K^n)$ and we define the {\em
 	regular dual space} $$\F Cl(S^1,\K^n)' = \{ \mu \in L(\F Cl(S^1,\K^n),\K) : \mu =
 \left< P , \cdot \right> \mbox{ for some } P \in \F Cl(S^1,\K^n) \}\; .$$
 We can adapt standard results described in section \ref{s:poisson} of Hamiltonian mechanics as follows:
let $f : \F Cl(S^1,\K^n)'
 \rightarrow B$  be a polynomial function of the type
$$f(\mu) = \sum_{k=0}^n a_k res_+(P^k) + \sum_{k=0}^n b_k res_-(P^k) =  res\left(\sum_{k=0}^n a_k P^k_+ +  b_k P^k_-\right)$$
 with $\mu = \left< P,. \right>.$ 
 In our picture, the decomposition $\F Cl(S^1,\K^n) = \F Cl_+(S^1,\K^n) \oplus \F Cl_-(S^1,\K^n)$ that we use extensively all along this work carry a residue trace on each component of the decomposition. These are these two residues, $res_+$ and $res_-,$ that replace ${\rm Res}$ in the constructions of section \ref{s:poisson}. Under these assumptions, we define the same way the functional derivative and the pairing $<.|.>$ of $\F Cl(S^1,\K^n)'$ with $\F Cl(S^1,\K^n).$
 The decomposition$\F Cl(S^1,\K^n) = \F Cl_D(S^1,\K^n) \oplus \F Cl_S(S^1,\K^n)$  allows us to consider a new Lie bracket on the
 regular dual space $\F Cl(S^1,\K^n)'$ given by
 $
 [ P , Q ]_{0} = [ P_{D} , Q_{D} ] - [ P_{S} , Q_{S} ] \; , $
  This bracket
 determines a new Poisson structure $\{ \, , \, \}_{0}$ on
 $\F Cl(S^1,\K^n)'$, simply by replacing the original Lie product for $[\;
 ,\;]_0$. Using again the non-degenerate
 pairing we get:
 
 \begin{Lemma}   \label{adler2}
 	Let $H : \F Cl(S^1,\K^n)' \rightarrow \K$ be a smooth function on $\F Cl(S^1,\K^n)'$ such that
 	\begin{equation}
 	\left< \, \mu \; \left| \; \left[ \frac{\delta H}{\delta \mu} \, ,
 	\, \cdot \, \right] \right> \right. \, = 0   \quad \quad \mbox{
 		for all } \mu \in \F Cl(S^1,\K^n)' \; .
 	\label{ad1}
 	\end{equation}
 	Then, as equations on $\F Cl(S^1,\K^n)$, the Hamiltonian equations of motion
 	with respect to the $\{ \, , \, \}_{0}$ Poisson structure of
 	$\F Cl(S^1,\K^n)''$ are
 	\begin{equation}
 	\frac{d\,P}{d\,t} = \left[ \left( \frac{\delta H}{\delta \mu} \right)_{+} \, , \, P \right] \; .    \label{lax2}
 	\end{equation}
 \end{Lemma}
 
 We now use some specific functions $H$. Let us recall the
 following results (see for example \cite{D} or the more recent
 review \cite{ER2013}):
 
 \begin{Proposition} \label{casimir}
 	We define the functions $\displaystyle H_{k}(L) = Trace \left(
 	(L^{k}) \right)$,
 	$k=1,2,3,\cdots ,$ for $L \in \F Cl(S^1,\K^n)$. Then, $\displaystyle
 	\frac{\delta H_{k}}{\delta L} = k L^{k-1}$. In particular, the functions $H_k$ satisfy $(\ref{ad1})$.
 \end{Proposition}
 
 Thus, we can apply Lemma \ref{adler2}. It yields:
 
 %
 %
 
 \begin{Proposition} \label{kpp}
 	Let us equip the Lie algebra $\F Cl(S^1,\K^n)$ with the non-degenerate pairing
 	$(a,b)\mapsto res(ab).$ Write $\F Cl(S^1,\K^n) = \F Cl_D(S^1,\K^n) \oplus \F Cl_S(S^1,\K^n)$
 	and consider the Hamiltonian functions
 	\begin{equation}
 	\mathcal{H}_{k}(\mu) = \frac{1}{k} \, res_W \left( (L^{k+1}) \right) \label{kpham}
 	\end{equation}
 	for $\mu = \left< L,.\right>$. The corresponding Hamiltonian equations of
 	motion with respect to the $\{ \, , \, \}_{0}$ Poisson structure
 	of $\F Cl(S^1,\K^n)'$ are
 	$
 	\frac{d L}{d t_{k}} = \left[ (L^{k})_{D} , L \right] \; .  $
 \end{Proposition}
 
 Following now \cite{EKRRR1995} we get the Hamiltonian formulation of the KP hierarchy with complex powers: 
 For this, we need to generalize the Gelfand-Dickii stricture  either to $$\mathcal{L}_\alpha = \left(\frac{d}{dx}\right)^\alpha + \F Cl^{\alpha-1}(S^1,\K^n)$$ or to $$\mathcal{L}'_\alpha = \left|\frac{d}{dx}\right|^\alpha + \F Cl^{\alpha-1}(S^1,\K^n). $$
 In both case, we specialize our computations to $\F Cl^\alpha_+(S^1,\K^n)$, and with $\F Cl^\alpha_-(S^1,\K^n)$ respectively, which both identify with $\Psi DO^\alpha(S^1,\K^n).$ Under these identifications, the computations described in \cite[pp 55--57]{EKRRR1995}:
 
 \begin{Theorem}
 	On $ \mathcal{L}_\alpha$ and on $\mathcal{L}_\alpha',$ the Hzmiltonian vector field associated to $H_k = \frac{\alpha}{k} res L^{k/\alpha}$
 	is $V = \left[L^{k/\alpha}_D,L\right].$
 	\end{Theorem}
 
	\section{Appendix: Proofs}

We collect in the Appendix all routine and technical prooofs which are often nothing but a straightforward verification. The end of the Appendix contains also technical proofs of some
theorems about KP hierarhies which are very similar ideologically of our ancient proofs from \cite{EKRRR1995}.

	\subsection{Proofs of section \ref{s:tech}}
	\begin{proof}[Lemma \ref{lemma:s'-pair}]
		Since $res$ is non degenerate then $(.;.)_{s'}$ is non degenerate. Moreover, identifying $\F Cl_+(S^1,V)$  and $\F Cl_-(S^1,V)$ as two copies of $\Psi DO(S^1,V)$, writing by $Tr$ the Adler trace on the latter one, we have that 
		$$
			res(A,s'(B)) =  Tr(A_+ B_-) + Tr(A_- B_+) =  Tr(B_+A_-) + Tr(B_-A_+) =  res(B,s'(A))
		$$which proves symmetry. 	
	\end{proof}
	\begin{proof}[Lemma \ref{lemma:s-pair}]
		Since $res$ is non degenerate then $(.;.)_{s}$ is non degenerate. Moreover, 
		\begin{eqnarray*}
			res(A,s(B)) & = & res((A_{ee}+A_{eo})+ (B_{ee}-B_{eo})) \\
			& = & res(A_{eo}B_{ee}) -  res(A_{ee}B_{eo})\\
			& = & -res(B,s(A))
		\end{eqnarray*} which proves skewsymmetry. 
		\begin{eqnarray*}
			res(A,s([B,C])) & = & res((A_{ee}+A_{eo})+ ([B,C]_{ee}-[B,C]_{eo})) \\
			& = & res(A_{eo}B_{ee}C_{ee})+ res(A_{eo}B_{eo}C_{eo}) -  res(A_{ee}B_{ee}C_{eo})- res(A_{ee}B_{eo}C_{ee}) \\ && -
			res(A_{eo}C_{ee}B_{ee}) - res(A_{eo} C_{eo}B_{eo}) +  res(A_{ee}C_{ee}B_{eo}]) + res(A_{ee}C_{eo}B_{ee}])
		\end{eqnarray*}
		while, with the same calculations, \begin{eqnarray*}res([A,C],s(B)) 
			& = & res(A_{eo}C_{ee}B_{ee}) + res(A_{ee}C_{eo}B_{ee})-  res(A_{ee}C_{ee}B_{eo}) - res(A_{eo}C_{eo}B_{eo})\\&  & - res(A_{eo}B_{ee}C_{ee}) - res(A_{ee}B_{ee}C_{eo})+  res(A_{ee}B_{eo}C_{ee}) + res(A_{eo}B_{eo}C_{eo})
		\end{eqnarray*}
		Let us investigate the same properties with $[.,.]_{\epsilon(D)}:$
		\begin{eqnarray*}
			res(A,s([B,C]_{\epsilon(D)})) & = & res((A_{ee}+A_{eo}) (-\epsilon(D)[B,C]_{ee}+ \epsilon(D)[B,C]_{eo})) \\
			& = & res(\epsilon(D)A_{eo}B_{ee}C_{eo})+ res(\epsilon(D)A_{eo}B_{eo}C_{ee})\\&  & -  res(\epsilon(D)A_{ee}B_{eo}C_{eo})- res(\epsilon(D)A_{ee}B_{ee}C_{ee}) \\ && -
			res(\epsilon(D)A_{eo}C_{eo}B_{ee}) - res(\epsilon(D)A_{eo} C_{ee}B_{eo})\\&  & +  res(\epsilon(D)A_{ee}C_{eo}B_{eo}) + res(\epsilon(D)A_{ee}C_{ee}B_{ee})
		\end{eqnarray*}
		while \begin{eqnarray*}res ([A,C],s(B)) 
			& = & res(\epsilon(D)A_{ee}C_{ee}B_{ee}) + res(\epsilon(D)A_{eo}C_{eo}B_{ee})\\&  &-  res(\epsilon(D)A_{eo}C_{ee}B_{eo}) - res(\epsilon(D)A_{ee}C_{eo}B_{eo})\\&  & - res(\epsilon(D)A_{ee}B_{ee}C_{ee}) - res(\epsilon(D)A_{eo}B_{ee}C_{eo})\\&  &+  res(\epsilon(D)A_{eo}B_{eo}C_{ee}) + res(\epsilon(D)A_{ee}B_{eo}C_{eo})
		\end{eqnarray*}
	\end{proof}
{ \begin{proof}[Theorem \ref{1.25}]
		First, $\F Cl_{ee}(S^1,V)$ is itself a subalgebra of $\F Cl(S^1,V)$ hence $( \F Cl_{ee}(S^1,V), [.,.])$ is a Lie subalgebra of $\F Cl_{ee}(S^1,V)$ on which $res(A,B)$ satisfies the same well-known properties: bilinear and symmetric. Moreover, it is well-known that $res$ is non-degenerate on $\F Cl(S^1,V).$ From \cite{KV1}, one can deduce by considering only formal operators that $\F Cl_{ee}(S^1,V)$ is isotropic for $res.$ Secondly, since $\epsilon(D)$ commutes with any element of $\F Cl(S^1,V),$ we have that \begin{eqnarray*}\forall A,B,C \in  \F Cl_{eo}(S^1,V), & & [A,[B,C]_{\epsilon(D)}]_{\epsilon(D)}+ [C,[A,B]_{\epsilon(D)}]_{\epsilon(D)} + [B,[C,A]_{\epsilon(D)}]_{\epsilon(D)}\\& & =  \epsilon(D)^2 \left([A,[B,C]]+ [C,[A,B]] + [B,[C,A]]\right)=0,\end{eqnarray*} which proves that $( \F Cl_{eo}(S^1,V), [.,.]_\epsilon(D))$ is a Lie bracket. Moreover, $A \mapsto \epsilon(D)A$ is a vector space isomorphism from  $ \F Cl_{eo}(S^1,V)$ to $ \F Cl_{ee}(S^1,V),$ which implies, with $\epsilon(D)^2 = 1,$ that $$\forall (A,B) \in  \F Cl_{eo}(S^1,V), res(AB) = res\left((\epsilon(D)A)(\epsilon(D)B)\right)$$
		and shows that $\F Cl_{eo}(S^1,V)$ is isotropic for $res(AB).$ 
	Let us finish the proof with invariance on $res.$ Invariance with respect to $[.,.]$  in $\F Cl(S^1,V)$ is well-known since $res$ is tracial. Again since $\epsilon(D)$ commutates, we have that $\forall(A,B) \in \F Cl(S^1,V), [A,B]_{\epsilon(D)} = [\epsilon(D)A,B] = [A,\epsilon(D)B]$ hence for $(A,B,C) \in \F Cl(S^1,V)^3,$
	\begin{eqnarray*}res([A,B]_{\epsilon(D)}C) & = & res([{\epsilon(D)}A,B]C) \\
		& = & res(B[C,{\epsilon(D)}A]) \\
		& = & - res (B[A,C]_{\epsilon(D)}). \end{eqnarray*}
	\end{proof}}
	\subsection{Proofs of section \ref{s:id}}
	\subsubsection{Proof of Theorem \ref{th:J1}}
	The operator $i\epsilon(D)$ commutes with any operator $u \in \mathcal{F}Cl(S^1,V).$ By the way, we simplify the relation that can be found e.g. in \cite{Mal1968} the following way:
$	\left[u, J_1v\right]+ 	\left[J_1u, v\right]  =  2 J_1[u,v]
$
	and 
	$$
		J_1\left[u,v\right] - J_1	\left[J_1u, J_1v\right]  =  J_1\left[u,v\right] - J_1^3 	\left[u, v\right]
		 =  2 J_1[u,v]
	$$
	Hence $$\left[u, J_1(v)\right]+ 	\left[J_1(u), v\right] = J_1\left[u, v\right] - J_1 	\left[J_1(u), J_1(v)\right]$$
	which proves integrability. 

\begin{Lemma}
We have $J_1(\mathcal{F}Cl_{ee}(S^1,V)) =\mathcal{F}Cl_{eo}(S^1,V) $
and $J_1(\mathcal{F}Cl_{eo}(S^1,V) ) =\mathcal{F}Cl_{ee}(S^1,V) $
\end{Lemma}

\begin{proof}
Since $J=i\epsilon(D) \circ (.)$ it follows from the composition rules between even-even and even-odd class already described.
\end{proof}
Identifying $\mathcal{F}Cl_{eo}(S^1,V)$ with $\epsilon(D) \mathcal{F}Cl_{ee}(S^1,V),$ we recall that $$ \mathcal{F}Cl(S^1,V) = \mathcal{F}Cl_{ee}(S^1,V) \oplus \epsilon(D) \mathcal{F}Cl_{ee}(S^1,V),$$
we get the complexification result.
	\subsubsection{Other proofs}
	\begin{proof}[Proposition \ref{prop:J1J2}]
		By straightforward computations, we check first that $s(\epsilon(D))=  -\epsilon(D). $ Then, since composition of symbols by $\epsilon(D)$ is only pointwise multiplication, we get, for  $a \in \mathcal{F}Cl(S^1,V),$
		\begin{eqnarray*}
			J_2 \circ J_1 (a)(x,\xi) & = & - s(\epsilon(D) \circ a) (x,\xi)\\
			& = & - \epsilon(D)(x,-\xi) \left(\sum_{k \in \Z} (-1)^ka_k(x,-\xi)\right)\hbox{ (pointwise multiplication)} \\
			& = & \epsilon(D)(x,\xi) s(a) (x,\xi) \hbox{ (pointwise multiplication)} \\
			& = & - i \epsilon(D)\circ (i  s(a)) (x,\xi)\\
			& = & -J_1 \circ J_2 (a) (x,\xi)
		\end{eqnarray*}  
	\end{proof}
	\begin{proof}[Theorem \ref{th:J2int}]
		
		\begin{eqnarray*}
			\left[u, J_2v\right]+ 	\left[J_2u, v\right] & = & i[u,v_{ee}] - i [u, v_{eo}] +i[u_{ee},v] - i [u_{eo},v] \\
			& = & i \left([u_{ee},v_{ee}] + [u_{eo},v_{ee}] - [u_{ee}, v_{eo}] - [u_{eo}, v_{eo}] \right. \\ && \left. + [u_{ee},v_{ee}] + [u_{ee},v_{eo}] -[u_{eo},v_{ee}] - [u_{eo},v_{eo}] \right)\\
			& = & 2i \left([u_{ee},v_{ee}]   - [u_{eo}, v_{eo}]  \right)\\
		\end{eqnarray*}
		and 
		\begin{eqnarray*}
			J_2\left[u,v\right] - J_2	\left[J_2u, J_2v\right] & = & i\left([u_{ee},v_{ee}] + [u_{eo}, v_{eo}]\right) -i\left([u_{ee},v_{eo}] + [u_{eo}, v_{ee}] \right)
			\\
			&  & +i\left([u_{ee},v_{ee}] + [-u_{eo}, -v_{eo}]\right) -i\left([u_{ee},-v_{eo}] + [-u_{eo}, v_{ee}] \right)\\
			& = & 2i\left([u_{ee},v_{ee}] + [u_{eo}, v_{eo}]\right)
		\end{eqnarray*}
	As a counter-example, let $X = f(x) \partial $ and let $Y = g(x) \partial$ be two vector fields over $S^1$ such that $[X,Y] \neq 0.$ Let $u = \epsilon(D) X \in \F Cl_{eo}(S^1,\R)$ and let $v = \epsilon(D) Y \in \F Cl_{eo}(S^1,\R).$ Then $ 	J_2\left[u,v\right] - J_2	\left[J_2u, J_2v\right] = 2i [X,Y]$ while  $\left[u, J_2v\right]+ 	\left[J_2u, v\right] = -2i [X,Y].$
	\end{proof}
\begin{proof}[Proposition \ref{prop:J1J3}]
	By straightforward computations, we check first that $s'(\epsilon(D))=  -\epsilon(D). $ Then, since $s'$ is a morphism of algebra, 
	\begin{eqnarray*}
		J_3 \circ J_1 (a) & = & - s'(\epsilon(D) \circ (a_+,a_-)) 
		= 
		-  (-a_-,a_+) \\
		& = & - \epsilon(D)\circ (-a_-,-a_+) 
		= 
		\epsilon(D)\circ s'(a_+,a_-)
		= 
		-J_1\circ J_3 (a)
	\end{eqnarray*}  
\end{proof}
\begin{proof}[Proposition \ref{prop:J3J2}]
	By straightforward computations, we check first that $ss'=  s's. $ Then,$J_2J_3 = J_3J_2.$
\end{proof}

\begin{proof}[Theorem \ref{th:J3int}]
	
	$
		\left[u, J_3v\right]+ 	\left[J_3u, v\right]  =  i\left([u_+v_-]+ [u_-v_+], [u_+v_-]+ [u_-v_+]\right)
	$
	and 
	$
		J_2\left[u,v\right] - J_2	\left[J_2u, J_2v\right]  =  i\left([u_{+},v_{+}] + [u_{-}, v_{-}],[u_{+},v_{+}] + [u_{-}, v_{-}] \right)
	$
	As a counter-example, let $X = f(x) \partial $ and let $Y = g(x) \partial$ be two vector fields over $S^1$ such that $[X,Y] \neq 0.$ Let $u = X_+ \in \F Cl_{+}(S^1,\R)$ and let $v = Y_+ \in \F Cl_{+}(S^1,\R).$ Then $ 	J_3\left[u,v\right] - J_2	\left[J_3u, J_3v\right] = i [X,Y]_+$ while  $\left[u, J_3v\right]+ 	\left[J_3u, v\right] =0.$
\end{proof}

\begin{proof}[Proposition \ref{Prop:J4}]
	
	\begin{itemize}
		\item $J_4^2  =  J_1J_3J_1J_3  =  -J_3J_1J_1J_3 =  - Id$
		\item $J_2J_4  =  J_2J_1J_3  =  -J_1J_2J_3 =  -J_1J_3J_2 =  - J_4J_2$
	\item $
		J_1J_4  =  J_1J_1J_3  =  -J_1J_3J_1 
		 =  -J_3J_1$
	\end{itemize}
	\end{proof}
	\subsection{Proofs of section \ref{s:KP}}
	\begin{proof}[Proof of Proposition \ref{prop:R-Kn}]
		From $$\F Cl(S^1,\K^n) = \F Cl_{ee}(S^1,\K^n) \oplus \F Cl_{eo}(S^1,\K^n) = \F Cl_{ee}(S^1,\K^n) \oplus i\epsilon(D)\F Cl_{ee}(S^1,\K^n)$$  
		we get, for $A = A_{ee} + A_{eo}\in \F Cl_{ee}(S^1,\K^n) \oplus \F Cl_{eo}(S^1,\K^n),$ and for $k \in \Z,$
		\begin{eqnarray*} \sigma_k{A} & = & \sigma_k{(A_{ee})} + \sigma_k(A_{eo}) =  a_{k,ee} \frac{d}{dx}^k + a_{k,eo}i\epsilon(D) \frac{d}{dx} \\
			& = & \left(a_{k,ee}  + i\epsilon(D)a_{k,eo}\right) \partial^k
		\end{eqnarray*}
		(where $ (a_{k,ee},a_{k,eo}) \in C^\infty(S^1,M_n(\K))$)which ends the identification of $\F Cl(S^1,\K^n)$ with $\Psi DO(\mathcal{R})$. Since the order of partial symbols is conserved, we get the same identifications between $\F Cl_D(S^1,\K^n)$ and $DO(\mathcal{R})$, and between $\F Cl_S(S^1,\K^n)$ and $IO(\mathcal{R}).$
	\end{proof}
	 \begin{proof}[Proof of Theorem \ref{th:KPlambdamu}]
	 	{
	 	We analyze separately the equation on $\F Cl_+(S^1,\K)$ and on $\F Cl_-(S^1,\K). $ 
	 	Let us work on $\F Cl_+(S^1,\K).$ The map $\Phi_{1,0}$ pulls-back the KP hierarchy on $\Psi DO(S^1,\K)$ with initial value $L_0 \in \lambda \partial + \Psi DO(S^1,\K).$ When $\lambda \neq 1,$ the classical integration of the KP hierarchy is not achieved by the classical method. However, we use here the scaling first defined to our knowledge in \cite{Ma2013}. Let $q = \lambda^{-1}.$ Let $\tilde L_0 = q L_0.$
	 	Then $\tilde L_0 \in \lambda \partial + \Psi DO(S^1,\K),$ and there exists a Sato operator $S_0 \in 1 + \Psi DO^{-1}(S^1,\K)$ such that $\tilde L_0 = S_0 \partial S_0^{-1}$ and the KP-system with initial value $\tilde L_0$ has a unique solution $\tilde L.$ We define a $\lambda-$scaling in time:
	 	$ t_k \mapsto \lambda^k t_k.$
	 Following \cite{Ma2013}, $$\tilde L(t_1,t_2,...) \hbox{ is solution of (\ref{eq:KP}}) \Leftrightarrow L(t_1,t_2,..) = \lambda \tilde L(\lambda t_1, \lambda^2 t_2,...) \hbox{ is solution of (\ref{eq:KP}}).$$
	 The initial value of the solution $L$ is $L_0,$ which proves existence, uniqueness and smooth dependence of $L$ on $L_0.$
	 Then, we can push-forward the solution $L$ of (\ref{eq:KP})on $\Psi DO(S^1,\K)$ to the solution $L_+ = \Phi_{1,0}(L)$ on $\F Cl_+(S^1,\K).$
	 The same procedure holds to get the solution $L_-$ on $\F Cl_+(S^1,\K),$ replacing the constant $\lambda$ by the constant $\mu.$ The operator $L_+ + L_-$ furnishes the desired solution of (\ref{eq:KP}) on $\F Cl(S^1,\K),$ which is unique and smoothly dependent on the initial value by construction.}\end{proof}  	
 	\begin{proof}[Proof of Theorem \ref{th:KPepsilon}]
 		Let us transform slightly Equation \ref{eq:KPepsilon} for $k \in \N^*$:
 		$$\frac{d L}{d t_{k}} = \epsilon^{k}\left[ (L^{k})_{D} , L \right] \Leftrightarrow \epsilon  \frac{d L}{d t_{k}} = \epsilon^{k+1}\left[ (L^{k})_{D} , L \right] \Leftrightarrow   \frac{d (\epsilon L)}{d t_{k}} = \left[ ((\epsilon L)^{k})_{D} , (\epsilon L) \right]
 		$$
 		By the way,the field of operators $\epsilon L,$ with initial value $\epsilon L_0 \in \partial_{a\lambda, -a\mu} +  \F Cl^{-1}(S^1,\K^n),$ is the unique solution of the KP hierarchy (\ref{eq:KP}). Moreover, the map  $L \mapsto \epsilon L$ is smooth, biunivoque, with smooth inverse, which ends the proof. 
 	\end{proof}
\begin{proof}[Proof of Theorem \ref{th:KPalpha}]
	Let us first analyze case (\ref{eq:complex1}). Then $$L_0^{1/\alpha} = \exp\left(\frac{1}{\alpha}\log L\right) \in \frac{d}{dx} + \F Cl^0(S^1,\K).$$
	We can then adapt \cite{RS1981} and define the dressing operator
	$$ U = \exp\left(\sum_{k \in \N^*} t_k (L_0^{1/\alpha})^k\right)$$
	that we decompose into 
	$ U = U_+ + U_-$
	where  $$ U_\pm = \exp\left(\sum_{k \in \N^*} t_k (L_0^{1/\alpha})_\pm^k\right).$$
	Working independently on the two components $\F Cl_\pm(S^1,\K),$
	Mulase factorization holds, 
	$U_\pm = S^{-1}_\pm Y_\pm$
	and setting $L_\pm = Y_\pm (L_0)_\pm Y^{-1}_\pm,$
	which implies  that $\forall k \in \N^*, \quad L_\pm^{k/\alpha} = Y_\pm (L_0)_\pm^{k/\alpha} Y^{-1}_\pm.$
	We moreover have that $\forall \beta \in \C^*, \quad L_0^{\beta} U = U L_0^{\beta}$
	which implies that 
	\begin{eqnarray*}
		L_\pm& =& Y_\pm (L_0)_\pm Y^{-1}_\pm = S_\pm S_\pm^{-1} Y_\pm (L_0)_\pm Y^{-1}_\pm S_\pm S^{-1}_\pm = S_\pm (L_0)_\pm S^{-1}_\pm
	\end{eqnarray*}
	and similarily 
$
		L_\pm^{k/\alpha} = Y_\pm (L_0)^{k/\alpha}_\pm Y^{-1}_\pm = S_\pm (L_0)^{k/\alpha}_\pm S^{-1}_\pm
$
	We can now differentiate $U$ 
	\begin{eqnarray*}
		\frac{d U_\pm}{dt_k} & = & (L_0)^{k/\alpha}_\pm U_\pm  =  - S^{-1}_\pm \frac{d S_\pm}{dt_k}S^{-1}_\pm Y_\pm + S^{-1}_\pm \frac{d Y_\pm}{dt_k} \end{eqnarray*}
	which implies that 
	$$ S_\pm (L_0)^{k/\alpha}_\pm S^{-1}_\pm = - \frac{d S_\pm}{dt_k}S^{-1}_\pm + \frac{d Y_\pm}{dt_k} Y_\pm.$$
	By the way, $ \left(L^{k/\alpha}_\pm\right)_D = \frac{d Y_\pm}{dt_k} Y_\pm$
	and hence 
	\begin{eqnarray*}
		\frac{d L_\pm}{dt_k} & = & \frac{d  Y_\pm (L_0)_\pm Y^{-1}_\pm}{dt_k} =  \frac{d  Y_\pm}{dt_k} (L_0)_\pm Y^{-1}_\pm -   Y_\pm (L_0)_\pm Y^{-1}_\pm\frac{dY_\pm}{dt_k}Y^{-1}_\pm \\
		& = & \left(L^{k/\alpha}_\pm\right)_DY_\pm (L_0)_\pm Y^{-1}_\pm - Y_\pm (L_0)_\pm Y^{-1}_\pm\left(L^{k/\alpha}_\pm\right)_D \\
		& = & \left[ \left(L^{k/\alpha}_\pm\right)_D,Y_\pm (L_0)_\pm Y^{-1}_\pm\right]
	\end{eqnarray*}
	Gathering the $\pm$ parts, we get that $L=L_+ +L_-$ is a solution of (\ref{eq:KPalpha}) with initial condition (\ref{eq:complex1}). 
	Let us now deal with initial condition (\ref{eq:complex2}). In that case, 
	$$L_0^{1/\alpha} \in |D| + \F Cl_0(S^1,\K) =  i\epsilon(D)\left(\frac{d}{dx} + \F Cl_0(S^1,\K)\right)  =  i\epsilon(D)\frac{d}{dx} + \F Cl_0(S^1,\K).$$
	Let $U$ be the dressing operator with respect to $L_0 $ along the lines of the previous computations, that we decompose into the $$U=U_+ + U_-$$ in the  $\F Cl_+(S^1,\K)$ and $\F Cl_-(S^1,\K)-$components. Then $U_+$ and $U_-$ decompose in the Mulase decomposition and we can re-construct two operators $S$ and $Y$ in $\F Cl(S^1,\K)$ from this decomposition. We set $ L = Y L_0 Y^{-1} = S L_0 S^{-1}.$
	Then, with the same computation as before, 
	\begin{eqnarray*}
		\frac{d L}{dt_k} & = & \frac{d  Y L_0 Y^{-1}}{dt_k} =  \frac{d  Y}{dt_k} L_0 Y^{-1} -   Y L_0 Y^{-1}_\pm\frac{dY}{dt_k}Y^{-1} \\
		& = & \left(L^{k/\alpha}\right)_DY L_0 Y^{-1} - Y L_0 Y^{-1}\left(L^{k/\alpha}\right)_D  =  \left[ \left(L^{k/\alpha}\right)_D,Y L_0 Y^{-1}\right]
	\end{eqnarray*}
\end{proof}
	
\end{document}